\documentclass[11pt,oneside]{article}

\usepackage[english]{babel}
\usepackage{amsmath}
\usepackage{amssymb}
\usepackage{amsthm}
\usepackage{mathtools}
\usepackage{xspace}
\usepackage{tikz}

\usepackage{fullpage}

\theoremstyle{plain}
\newtheorem{conjecture}             {Conjecture}
\newtheorem{corollary}  [conjecture]{Corollary}
\newtheorem{claim}                  {Claim}
\newtheorem{definition} [conjecture]{Definition}
\newtheorem{lemma}      [conjecture]{Lemma}
\newtheorem{notation}   [conjecture]{Notation}
\newtheorem{problem}    [conjecture]{Problem}

\newtheorem{remark}     [conjecture]{Remark}

\newtheorem{theorem}    [conjecture]{Theorem}

\newcommand{\cC}{\ensuremath{\mathcal{C}}}
\newcommand{\cI}{\ensuremath{\mathcal{I}}}
\newcommand{\cL}{\ensuremath{\mathcal{L}}}
\newcommand{\cX}{\ensuremath{\mathcal{X}}}

\newcommand{\IN}{\ensuremath{\mathbb{N}}}
\newcommand{\IR}{\ensuremath{\mathbb{R}}}

\newcommand{\floor} [1]{\ensuremath{{\lfloor #1 \rfloor}}}

\newcommand{\angleb}[1]{\ensuremath{\langle #1\rangle}}
\newcommand{\eref}  [1]{(\ref{#1})}

\newcommand{\dg} {\ensuremath{d}}
\newcommand{\ra} {\ensuremath{\rightarrow}}
\newcommand{\lra}{\ensuremath{\longrightarrow}}

\newcommand{\PT}{\textrm{P}}
\newcommand{\NP}{\textrm{NP}}
\newcommand{\QP}{\textrm{QP}}
\newcommand{\DTIME}{\textrm{DTIME}}

\newcommand{\LC} {\textsc{Label-Cover}\xspace}
\newcommand{\LCM}{\textsc{Label-Cover-Max}\xspace}

\title{
  Hardness Results for Approximate Pure Horn CNF Formulae Minimization\thanks{
    To appear in the special issue of the Annals of Mathematics and Artificial Intelligence dedicated to
    ISAIM-2012.
    The authors gratefully acknowledge the partial support by NSF grants IIS 0803444 and by CMMI 0856663.
    The second author also gratefully acknowledge the partial support by the joint CAPES (Brazil)/Fulbright (USA) 
    fellowship process BEX-2387050/15061676.
  }
}

\author{
  Endre Boros \qquad Aritanan Gruber\\[4mm]
  RUTCOR and MSIS, Rutgers University, 100 Rockafeller Road\\
  Piscataway, NJ, 08854, USA.\\
  \{boros, agruber\}@rutcor.rutgers.edu
}

\begin{document}
\maketitle

\begin{abstract}
  We study the hardness of approximation of clause minimum and literal minimum representations of pure Horn
  functions in $n$ Boolean variables. We show that unless $\PT=\NP$, it is not possible to approximate in polynomial 
  time the minimum number of clauses and the minimum number of literals of pure Horn CNF representations to 
  within a factor of $2^{\log^{1-o(1)} n}$. This is the case even when the inputs are restricted to pure Horn 
  3-CNFs with $O(n^{1+\varepsilon})$ clauses, for some small positive constant $\varepsilon$. Furthermore, we 
  show that even allowing sub-exponential time computation, it is still not possible to obtain constant factor 
  approximations for such problems unless the Exponential Time Hypothesis turns out to be false.
  \\[4mm]
  \\[4mm]
\end{abstract}

\section{Introduction}\label{sec:intro}
Horn functions constitute a rich and important subclass of Boolean functions and have many applications in artificial 
intelligence, combinatorics, computer science, and operations research. Furthermore, they possess some nice structural 
and algorithmic properties. An example of this claim is found on the satisfiability problem (SAT) of formulae in
conjunctive normal form (CNF): while it is one of the most famous \NP-complete problems for general Boolean CNFs
(Cook~\cite{C71}; see also: Arora and Barak~\cite{AB09}, Garey and Johnson~\cite{GJ79}), it can be solved in linear time 
in the number of variables plus the length of the Horn CNF formula being considered (Dowling and Gallier~\cite{DG84}, 
Itai and Makowsky~\cite{IM87}, and Minoux~\cite{M88}), where length in this context means the number of literal 
occurrences in the formula (i.e., multiplicities are taken into account).

The problem of finding short Horn CNF representations of Horn functions specified through Horn CNFs has received some
considerable attention in the literature since it has an intrinsic appeal stemming from both theoretical and practical 
standpoints. The same can be said about some special cases, including Horn 3-CNFs, i.e., the ones in which each 
clause has at most three literals.

Two of the common measures considered are the number of clauses and the number of literal occurrences (henceforth, 
number of literals). The \NP-hardness of minimizing the number of clauses in Horn CNFs was first proved in a slightly 
different context of directed hypergraphs by Ausiello et al.~\cite{ADS86}. It was later shown to also hold for the 
case of pure Horn 3-CNFs by Boros et al.~\cite{BCK13}. The \NP-hardness of minimizing the number of literals in Horn 
CNFs was shown by Maier~\cite{M80}. A simpler reduction was found later by Hammer and Kogan~\cite{HK93}. Boros et 
al.~\cite{BCK13} showed \NP-hardnes for the pure Horn 3-CNFs case.

The attention then shifted to trying to approximate those values. Hammer and Kogan~\cite{HK93} showed that for a pure 
Horn function on $n$ variables, it is possible to approximate the minimum number of clauses and the minimum number of 
literals of a pure Horn CNF formula representing it to within factors of $n-1$ and $\binom{n}{2}$, respectively. 
For many years, this was the only result regarding approximations. Recently, a super-logarithmic hardness 
of approximation factor was shown by Bhattacharya et al.~\cite{BDMT10} for the case of minimizing the number of clauses
for general Horn CNFs. We provide more details on this shortly. 

Another measure for minimum representations of Horn functions concerns minimizing the number of source sides, grouping 
together all clauses with the same source set. Maier~\cite{M80} and Ausiello et al.~\cite{ADS86} showed that such 
minimization can be accomplished in polynomial time (see also Crama and Hammer~\cite{CH11}). While this measure is 
sometimes used in practice, providing reasonably good results, we consider it an important intelectual quest, following
Bhattacharya et al.~\cite{BDMT10}, to try to precisely understand the hardness of the other two measures.

In this work, we focus on the hardness of approximating short pure Horn CNF/3-CNF representations of pure Horn functions,
where pure means that each clause has exactly one positive literal (definitions are provided in Section~\ref{sec:prelim}). 
More specifically, we study the hardness of approximating 
\begin{itemize}
\item[$\bullet$] the minimum number of clauses of pure Horn functions specified through pure Horn CNFs,
\item[$\bullet$] the minimum number of clauses and the minimum number
  literals of pure Horn functions specified through pure Horn 3-CNFs,
\end{itemize}
when either polinomial or sub-exponential computational time is available.

Below, we present pointers to previous work on the subject, discuss our results in more details, and mention
the main ideas behind them. We then close this section with the organization of the rest of the paper.

\subsection*{Previous Work}
The first result on hardness of approximation of a shortest pure Horn CNF representations was provided in Bhattacharya 
et al.~\cite{BDMT10}. Specifically, it was shown that unless $\NP\subseteq\QP=\DTIME(n^{\textrm{polylog}(n)})$, that is, 
unless every problem in $\NP$ can be solved in quasi-polynomial deterministic time in the size of the input's representation, 
the minimum number of clauses of a pure Horn function on $n$ variables specified through a pure Horn CNF formula cannot be 
approximated in polynomial (depending on $n$) time to within a factor of $2^{\log^{1-\varepsilon}n}$, for any constant 
$\varepsilon>0$ small enough.

This result is based on a gap-preserving reduction from a fairly well known network design problem introduced by
Kortsarz~\cite{K01}, namely, \textsc{MinRep} and has two main components: a gadget that associates to every 
\textsc{MinRep} instance $M$ a pure Horn CNF formula $h$ such that the size of an optimal solution to $M$ is 
related to the size of a clause minimum pure Horn CNF representation of $h$, and a gap amplification device that 
provides the referred gap. Despite being both necessary to accomplish the result, each component works in a rather 
independent way.

Inspired by the novelty of their result and by some characteristics of their reduction, we are able to further advance
the understanding of hardness of approximation of pure Horn functions. We discuss how we strength their result below.

\subsection*{Our Results and Techniques}
Our strengthening of the result of Bhattacharya et al.~\cite{BDMT10} can be summarized as follows: the hardness of 
approximation factor we present is stronger, the complexity theoretic assumption we use for polynomial time solvability 
is weaker, and the class of CNF formulae to which our results apply is smaller. We are also able to derive further 
non-approximability results for sub-exponential time solvability using a different complexity theoretic hypothesis.

In more details, for a pure Horn function $h$ on $n$ variables, we show that unless $\PT=\NP$, the minimum number of 
clauses in a prime pure Horn CNF representation of $h$ and the minimum number of clauses and literals in a prime pure 
Horn 3-CNF representation of $h$ cannot be approximated in polynomial (depending on $n$) time to within factors of 
$2^{\log^{1-o(1)} n}$
even when the inputs are restricted to pure Horn CNFs and pure Horn 3-CNFs with $O(n^{1+\varepsilon})$ clauses, for some 
small constant $\varepsilon>0$. It is worth mentioning that $o(1)\approx(\log\log n)^{-c}$ for some constant $c\in(0,1/2)$ 
in this case. 

After that, we show that unless the Exponential Time Hypothesis introduced by Impagliazzo and Paturi~\cite{IP01} 
is false, it is not possible to approximate the minimum number of clauses and the minimum number of literals of a 
prime pure Horn 3-CNF representation of $h$ in time $\exp(n^\delta)$, for some $\delta\in(0,1)$, to within 
factors of $O(\log^\beta n)$ for some small constant $\beta>0$. Such results hold even when the inputs are restricted 
to pure Horn 3-CNFs with $O(n^{1+\varepsilon})$ clauses, for some small constant $\varepsilon>0$. Furthermore, we also
obtain a hardness of approximation factor of $O(\log n)$ under slightly more stringent, but still sub-exponential
time constraints. We would like to point out that our techniques leave open the problem of determining hardness of 
approximation factors when $exp(o(n))$ computational time is available. We conjecture, however, that constant factor 
approximations are still not possible in that case.

The main technical component of our work is a new gap-preserving reduction\footnote{Our reduction is actually 
approximation-preserving, but we do not need or rely on such characteristic. See Vazirani~\cite{V01} or 
Williamson and Shmoys~\cite{WS11} for appropriate definitions.} from a graph theoretical problem called \LC 
(see Section~\ref{sec:LC} for its definition) to the problem of determining the minimum number of clauses in 
a pure Horn CNF representation of a pure Horn function. We show that our reduction has two independent parts: 
a core piece that forms an exclusive component (Boros et al.~\cite{BCKK10}) of the function in question and 
therefore, can be minimized separately; and a gap amplification device which is used to obtain the hardness 
of approximation factor. It becomes clear that the same principle underlies Bhattacharya et al.~\cite{BDMT10}.
The hardness of approximation factor comes (after calculations) from a result of Dinur and Safra~\cite{DS04} 
on the hardness of approximating certain \LC instances.

We then introduce some local changes into our reduction that allow us to address the case in which the 
representation is restricted to be a pure Horn 3-CNF. Namely, we introduce extra variables in order to
\emph{cubify} clauses whose degree is larger than three, that is, to replace each of these clauses by a 
collection of degree two or degree three clauses that provide the same logical implications. This is done 
for two families of high degree clauses and for each family we use a different technique: a \emph{linked-list} 
inspired transformation that is used on the classical reduction from SAT to 3-SAT instances (Garey and 
Johnson~\cite{GJ79}), and a \emph{complete binary tree} type transformation. The latter type 
is necessary to prevent certain shapes of prime implicates in minimum clause representations that would 
render the gap-amplification device innocuos. From this modified reduction, we are also able to derive in a 
straight-forward fashion a hardness result for determining the minimum number of literals of pure Horn 
functions represented by pure Horn 3-CNFs.

At this point we should mention that our reduction is somewhat more complicated than the one given in  
Bhattacharya et al.~\cite{BDMT10}. While we could adapt their reduction and obtain the same hardness of
approximation factor we present in the case of pure Horn CNF representations (as based in our 
Lemma~\ref{lem:tight-tcover} we can argue that the \LC and the \textsc{MinRep} problems are equivalent), 
the more involved gadget we use is paramount in extending the hardness result for the pure Horn
3-CNF case. Loosely speaking, the simple form of their reduction does not provide enough room to
correctly shape those prime implicates in minimum pure Horn 3-CNF representations that we mentioned
addressing in ours. In this way, the extra complications are justified.

Finally, using newer and slightly different results on the hardness of approximation of certain \LC
instances (Moshkovitz and Raz~\cite{MR-10-JACM}, Dinur and Harsha~\cite{DH13}) in conjunction with 
the Exponential Time Hypothesis~\cite{IP01}, we are able to show (also after calculations) that it 
might not be possible to obtain constant factor approximations for the minimum number of clauses and 
literals in pure Horn 3-CNF representations.

\subsection*{Outline}
The remainder of the text is organized as follows. We introduce some basic concepts about pure Horn functions in 
Section~\ref{sec:prelim} and present the problem on which our reduction is based in Section~\ref{sec:LC}. The reduction 
to pure Horn CNFs, its proof of correctness, and the polynomial time hardness of approximation result are shown in 
Section~\ref{sec:reduction}. In Section~\ref{sec:3-cnf}, we extend that result to pure Horn 3-CNF formulae and address 
the case of minimizing the number of literals. In Section~\ref{sec:subexp} we show that sub-exponential time availability 
gives smaller but still super-constant hardness of approximation factors. We then offer some final thoughts in 
Section~\ref{sec:consider}. A short conference version of this article appeared in Boros and Gruber~\cite{BG12}.

\section{Preliminaries}\label{sec:prelim}
In this section we succinctly define the main concepts and notations we shall use later on. For an almost comprehensive 
exposition, consult the book on Boolean Functions by Crama and Hammer~\cite{CH11}.

A mapping $h:\{0,1\}^n\ra\{0,1\}$ is called a \emph{Boolean function} on $n$ propositional variables. Its set of variables is
denoted by $V_h:=\{v_1,\ldots,v_n\}$. A \emph{literal} is a propositional variable $v_i$ (positive literal) or its negation 
$\bar{v}_i$ (negative literal). 

An elementary disjunction of literals
\begin{equation}\label{eq:clause}
  C=\bigvee_{i\in I}\bar{v}_i\vee\bigvee_{j\in J}v_j,
\end{equation}
with $I,J\subseteq V_h$ is a \emph{clause} if $I\cap J=\emptyset$. The set of variables it depends upon is 
$\mathsf{Vars}(C):=I\cup J$ and its \emph{degree} or \emph{size} is given by $\mathsf{deg}(C):=|I\cup J|$. 
It is customary to identify a clause $C$ with its set of literals.

\begin{definition}\label{dfi:pure-horn}
  A clause $C$ as in \eref{eq:clause} is called \emph{pure} or \emph{definite Horn} if $|J|=1$. For a pure Horn clause $C$, 
  the positive literal $v\in J$ is called its \emph{head} and $S=C\setminus J$ is called its \emph{subgoal} or \emph{body}. 
  To simplify notation, we sometimes write $C$ simply as $\bar{S}\vee v$ or as the implication $S\lra v$.
\end{definition}	

\begin{definition}\label{dfi:horn-cnf}
  A conjunction $\Phi$ of pure Horn clauses is a \emph{pure Horn formula in Conjunctive Normal Form} (for short, pure Horn CNF).
  In case every clause in $\Phi$ has degree at most three, the CNF is a 3-CNF. A Boolean function $h$ is called \emph{pure Horn} 
  if there is a pure Horn CNF formula $\Phi\equiv h$, that is, if $\Phi(v)=h(v)$ for all $v\in\{0,1\}^n$.
\end{definition}

Let $\Phi=\bigwedge_{i=1}^m C_i$ be a pure Horn CNF representing a pure Horn function $h$. We denote by
\[
  |\Phi|_c:=m\qquad \text{and}\qquad |\Phi|_l:=\sum_{i=1}^m\mathsf{deg}(C)
\]
the numbers of clauses and literals of $\Phi$, respectively. We say that $\Phi$ is a clause
(literal) minimum representation of $h$ if $|\Phi|_c\leq|\Psi|_c$ ($|\Phi|_l\leq|\Psi|_l$) for every other pure Horn CNF 
representation $\Psi$ of $h$. With this in mind, we define 
\begin{align*}
  \tau(h)    & := \min\{|\Phi|_c : \Phi\textrm{ is a pure Horn CNF representing }h\},\ \text{and}\\
  \lambda(h) & := \min\{|\Phi|_l \,: \Phi\textrm{ is a pure Horn CNF representing }h\}.
\end{align*}

\begin{problem}\label{pro:horn-min}
  The clause (literal) pure Horn CNF minimization problem consists in determining $\tau(h)$ ($\lambda(h)$) when
  $h$ is given as a pure Horn CNF. Similar definitions hold for the pure Horn 3-CNF case.
\end{problem}

A clause $C$ as in \eref{eq:clause} is an \emph{implicate} of a Boolean function $h$ if for all $v\in\{0,1\}^n$ it holds that 
$h(v)=0$ implies $C(v)=0$. An implicate is \emph{prime} if it is inclusion-wise minimal with respect to its set of literals. 
The set of prime implicates of $h$ is denoted by $\cI^p(h)$. It is known (cf. Hammer and Kogan~\cite{HK92}) that prime implicates 
of pure Horn functions are pure Horn clauses. A pure Horn CNF $\Phi$ representing $h$ is \emph{prime} if its clauses are prime 
and is \emph{irredundant} if the pure Horn CNF obtained after removing any of its clauses does not represent $h$ anymore. 
Let us note that a clause minimun representation may involve non prime implicates, though it is always irredundant. As 
Hammer and Kogan~\cite{HK92} pointed out any Horn CNF can be reduced in polynomial time to an equivalent prime and irredundant
CNF. In the sequel we shall assume all CNFs considered, including the clause minimum ones, to be prime and irredundant.

Let $C_1$ and $C_2$ be two clauses and $v$ be a variable such that $v\in C_1$, $\bar{v}\in C_2$, and $C_1$ and $C_2$ have no 
other complemented literals. The \emph{resolvent} of $C_1$ and $C_2$ is the clause
\[
  R(C_1,C_2):=(C_1\setminus\{v\})\cup(C_2\setminus\{\bar{v}\})
\]
and $C_1$ and $C_2$ are said to be \emph{resolvable}. It is known (e.g. Crama and Hammer~\cite{CH11}) that if $C_1$ and $C_2$ 
are resolvable implicates of a Boolean function $h$, then $R(C_1,C_2)$ is also an implicate of $h$. Naturally, the resolvent 
of pure Horn clauses is also pure Horn.

A set of clauses $\cC$ is \emph{closed under resolution} if for all $C_1,C_2\in\cC$, $R(C_1,C_2)\in\cC$. The \emph{resolution 
closure} of $\cC$, $R(\cC)$, is the smallest set $\cX\supseteq\cC$ closed under resolution. For a Boolean function $h$, let 
$\cI(h):=R(\cI^p(h))$. 

Let us note that the set of all implicates of a Horn function $h$ may, in principle, contain clauses involving arbitrary 
other variables, not relevant for $h$. To formulate proper statements one would need to make sure that such redundancies 
are also handled, which complicates the formulations. To avoid such complications, we focus on $\cI(h)$ in the sequel, 
which is completely enough to describe all relevant representations of $h$.

\begin{definition}\label{dfi:FC}
  Let $\Phi$ be a pure Horn CNF representing a pure Horn function $h$ and let $Q\subseteq V_h$. The 
  \emph{Forward Chaining} of $Q$ in $\Phi$, denoted by $F_\Phi(Q)$, is defined by the following algorithm. 
  Initially, $F_\Phi(Q)=Q$. As long as there is a pure Horn clause $\bar{S}\vee v$ in $\Phi$ such that 
  $S\subseteq F_\Phi(Q)$ and $v\not\in F_\Phi(Q)$, add $v$ to $F_\Phi(Q)$. Whenever a variable $v$ is added 
  to $F_\Phi(Q)$, we say that the corresponding clause $\bar{S}\vee v$ was \emph{trigged}. 
\end{definition}

The result below is pivotal in our work. It tells us that we can make inferences about a pure Horn function 
$h$ using any of its pure Horn CNF representations.

\begin{lemma}[Hammer and Kogan~\cite{HK93}]\label{lem:horn-FC}
  Two distinct pure Horn CNFs $\Phi$ and $\Psi$ represent the same pure Horn function $h$ if and only if 
  $F_\Phi(U)=F_\Psi(U)$, for all $U\subseteq V_h$. Consequently, we can do Forward Chaining in $h$, which we 
  denote by $F_h(\cdot)$, through the use of any of $h$'s representations.
\end{lemma}

The following definitions and lemma concerning exclusive sets of clauses are useful when decomposing and studying structural 
properties of Boolean functions.

\begin{definition}[Boros et al.~\cite{BCKK10}]\label{dfi:exclusive-set}
  Let $h$ be a Boolean function and $\cX\subseteq\cI(h)$ be a set of clauses. $\cX$ is an \emph{exclusive set of clauses of}
  $h$ if for all resolvable clauses $C_1,C_2\in\cI(h)$ it holds that: $R(C_1,C_2)\in\cX$ implies $C_1\in\cX$ and $C_2\in\cX$.
\end{definition}

An example of an exclusive set of clauses is given by the set of pure Horn implicates of a Horn function: it is not hard to 
see that if a resolvent is a pure Horn clause, then both the resolvable clauses must also be pure Horn.

\begin{definition}[Boros et al.~\cite{BCKK10}]\label{dfi:exclusive-comp}
  Let $\cX\subseteq\cI(h)$ be an exclusive set of clauses for a Boolean function $h$ and let $\cC\subseteq\cI(h)$ be such that
  $\cC\equiv h$. The Boolean function $h^{}_\cX=\cC\cap\cX$ is called the $\cX$-component of $h$.
\end{definition}

The following claim justifies the use of ``the'' in the previous definition.

\begin{lemma}[Boros et al.~\cite{BCKK10}]\label{lem:subfunc-switch}
  Let $\cC_1,\cC_2\subseteq\cI(h)$, $C_1\neq C_2$, such that $\cC_1\equiv\cC_2\equiv h$ and let $\cX\subseteq\cI(h)$ 
  be an exclusive set of clauses. Then $\cC_1\cap\cX\equiv\cC_2\cap\cX$ and in particular $(\cC_1\setminus\cX)\cup(\cC_2\cap\cX)$ 
  also represents $h$.
\end{lemma}

The above lemma is a particularly useful and important tool in our work. Loosely speaking, it means that once we are
able to identify an exclusive component $g$ of a function $h$, we can separately study $g$. Moreover, we can draw 
conclusions about $g$ using any of its representations (even alternate between distinct representations as 
convenient) and then reintegrate the acquired knowledge into the analysis of $h$.

The Forward Chaining procedure provides us with a convenient way of identifying exclusive families for pure Horn 
functions, as stated in the next lemma. This result appeared recently in Boros et al.~\cite{BCK13}. As we make
explicit use of it in the analysis of our construction, we decided to include its proof below for completeness.

\begin{lemma}[Boros et al.~\cite{BCK13}]\label{lem:FC-exclusive-set}
  Let $\Phi$ be a prime pure Horn CNF representing the function $h$, let $W\subseteq V_h$ be such that 
  $F_\Phi(W)=W$, and define the set
  \[
    \cX(W) := \{C\in\cI(h): \mathsf{Vars}(C)\subseteq W\}\label{e61}.
  \]
  Then $\cX(W)$ is an exclusive family for $h$.
\end{lemma}
\begin{proof}
  Let $W$ be as specified in the lemma's statement and suppose there are clauses $C_1, C_2\in\cI(h)$ such that 
  $R(C_1,C_2)\in\cX(W)$ but $\{C_1,C_2\}\not\subseteq\cX(W)$. By the definitions of $\cX(W)$ and of resolution, 
  all but one of the variables in $C=C_1\cup C_2$ must belong to $W$ and that variable, say $v\in C\setminus W$, 
  is precisely the variable upon which $C_1$ and $C_2$ are resolvable, that is, $v$ occurs as head in one of those 
  clauses. Now, since for the same input the forward chaining outcome is independent of the pure Horn representation 
  of $h$, it follows that $v\in F_\Phi(W)\neq W$, a contradiction.
\end{proof}

When showing hardness results, we make use of the standard asymptotic notations: $O(\cdot)$, $\Theta(\cdot)$, 
$\Omega(\cdot)$, $o(\cdot)$, and $\omega(\cdot)$. Namely, \emph{big-oh}, \emph{big-theta}, \emph{big-omega},
\emph{little-oh} or \emph{omicron}, and \emph{little-omega}, respectively. For definitions and examples of use, 
consult e.g. the book by Cormen et al.~\cite{CLRS09}. 

We say a function $f:\IN\to\IN$ (or $f:\IN\to\IR$) is \emph{quasi-linear} if $f\in O\bigl(n^{1+o(1)}\bigr)$ 
and it is \emph{nearly linear} if $f\in O\bigl(n^{1+\varepsilon}\bigr)$ for some constant $\varepsilon>0$ small 
enough; it is \emph{quasi-polynomial} or \emph{super-polynomial} if $f\in O\bigl(n^{\textrm{polylog}(n)}\bigr)$; 
and it is \emph{sub-exponential} if $f\in O\bigl(2^{o(n)}\bigr)$.

\section{The \LC Problem}\label{sec:LC}
The $\LC$ problem is a graph labeling promise problem formally introduced in Arora et al.~\cite{ABSS97} as a 
combinatorial abstraction of interactive proof systems (two-prover one-round in Feige et al.~\cite{FGLSS96} and 
Feige and Lov\'asz~\cite{FL92}, and probabilistically checkable in Arora and Motwani~\cite{AM98} and Arora and 
Safra~\cite{AS98}). It comes in maximization and minimization flavors (linked by a ``weak duality'' relation) and 
is probably the most popular starting point for hardness of approximation reductions. In this section, we 
introduce a minimization version that is best suited for our polynomial time results. Later in Section~\ref{sec:subexp}, 
when dealing with sub-exponential time results, we shall mention its maximization counterpart. 

\begin{definition}\label{dfi:LC-inst}
  A \emph{\LC instance} is a quadruple $\cL^{}_0=(G,L_0^{},L_0',\Pi^{}_0)$, where $G=(X,Y,E)$ is a 
  bipartite graph, $L_0^{}$ and $L_0'$ are disjoint sets of labels for the vertices in $X$ and $Y$, 
  respectively, and $\Pi^{}_0=(\Pi^0_e)^{}_{e\in E}$ is a set of constraints with each 
  $\Pi^0_e\subseteq L_0^{}\times L_0'$ being a non-empty relation of admissible pairs of labels 
  for the edge $e$. The \emph{size} of $\cL^{}_0$ is equal to $|X|+|Y|+|E|+|L^{}_0|+|L'_0|+|\Pi^{}_0|$.
\end{definition}

\begin{definition}\label{dfi:LC-total-cover}
  A \emph{labeling} for $\cL^{}_0$ is any function $f^{}_0:X\ra 2^{L_0^{}},Y\ra 2^{L_0'}\setminus\{\emptyset\}$ 
  assigning subsets of labels to vertices. A labeling $f^{}_0$ \emph{covers} an edge $(x,y)$ if for every label 
  $\ell'_0\in f^{}_0(y)$ there is a label $\ell_0^{}\in f^{}_0(x)$ such that $(\ell_0^{},\ell'_0)\in\Pi^0_{(x,y)}$. A 
  \emph{total-cover} for $\cL^{}_0$ is a labeling that covers every edge in $E$. $\cL^{}_0$ is said to be \emph{feasible} 
  if it admits a total-cover.
\end{definition}

Following Arora and Lund~\cite{AL96}, a way to guarantee that a \LC instance is feasible is by imposing an 
extra condition on it, namely, that there is a label $\ell_0'\in L_0'$ such that for each edge $e\in E$, there is
a label $\ell_0^{}\in L_0^{}$ with $(\ell_0^{},\ell_0')\in \Pi^0_e$. In this way, a labeling assigning $\ell_0'$ to
each vertex in $Y$ and the set $L_0^{}$ to each vertex in $X$ is clearly a total-cover. However, all \LC instances 
that we shall use are, by construction, guaranteed to be feasible. Therefore, we shall not dwell on such imposition
and shall consider only feasible \LC instances in the sequel.

\begin{definition}\label{dfi:LC-total-cover-cost}	
  For a total-cover $f^{}_0$ of $\cL^{}_0$, let $f^{}_0(Z):=\sum_{z\in Z}|f^{}_0(z)|$ with $Z\subseteq X\cup Y$.
  The \emph{cost} of $f^{}_0$ is given by $\kappa(f^{}_0):=f^{}_0(X)/|X|$ and $f^{}_0$ is said to
  be \emph{optimal} if $\kappa(f^{}_0)$ is minimum among the costs of all total-covers for $\cL^{}_0$. This
  minimum value we denote by $\kappa(\cL^{}_0)$.
\end{definition}

Observe that the feasibility of $\cL^{}_0$ implies that $1\leq\kappa(f^{}_0)\leq |L^{}_0|$, for any total-cover 
$f^{}_0$. Also, without loss of generality, we can assume that $G$ has no isolated vertices as they do not 
influence the cost of any labeling.

We now give an example of a \LC instance $\cL_0$. Let $U:=\{u_1,\ldots,u_n\}$ be a set of Boolean variables and let 
$\Phi:=\bigwedge_{i=1}^s\phi_i$ be a formula in CNF such that each clause of $\Phi$ depends on $k$ variables 
of $U$ (as in a variation of the satisfiability problem in which each clause has exactly $k$ literals).
For a clause $\phi\in\Phi$ and a variable $u\in U$, we write $u\in\phi$ whenever $\phi$ depends on $u$.

The bipartite graph $G=(X,Y,E)$ is constructed from $\Phi$ as follows. Let $X:=\{x_1,\ldots,x_{ks}\}$ have a
vertex for every occurrence of a variable in $\Phi$, and let $Y:=\{1,\ldots,s\}$ have a vertex for every
clause $\phi\in\Phi$. Let $X(u)\subseteq X$ denotes the set of vertices corresponding to the variable $u$,
and define
\[
  E:=\Big\{(x,j)\in X\times Y : x\in X(u) \text{ and } u\in\phi_j\Big\},
\]
that is, each vertex $j\in Y$ is connected to all occurrences of all variables in the clause $\phi_j$.

Define the label-sets as $L^{}_0:=\{0,1\}$ and $L'_0:=\{0,1\}^k$. For an edge $(x,j)\in E$, assume that 
$x\in X(u)$ and that $u$ is the $i$-th variable in $\phi_j$, and define
\[
  \Pi^0_{(x,j)}:=\Big\{(a_i,(a_1,\ldots,a_k)) : \phi_j(a_1,\ldots,a_k)=\textsf{True}\Big\},
\]
where $a_i\in L^{}_0$ and $(a_1,\ldots,a_k)\in L'_0$.

Now, it is not hard to see that in this case, there is a total-cover $f^{}_0$ with $\kappa(f^{}_0)=1$ if and only 
if $\Phi$ is satisfiable. Notice that choosing $k\geq 3$ establishes the \NP-completeness of the problem of 
deciding if an optimal total-cover for a given \LC instance has cost equal to one. 

The above example was adapted from Dinur and Safra~\cite{DS04}. Their original version is used in
the proof of Theorem~\ref{thm:LC-inapprox}. In that context however, $\Phi$ is a non-Boolean satisfiability 
instance produced by a probabilistic checkable proof system and the label-sets involved are larger (see 
Remark~\ref{rem:DS-inst-size} below).

\begin{definition}\label{dfi:tightness}
  A total-cover $f^{}_0$ is \emph{tight} if $f^{}_0(Y) := \sum_{y\in Y}|f^{}_0(y)| = |Y|$, i.e., if 
  for every $y\in Y$, it holds that $|f^{}_0(y)|=1$.
\end{definition}

\begin{lemma}\label{lem:tight-tcover}
  Every \LC instance $\cL^{}_0$ admits a tight, optimal total-cover.
\end{lemma}
\begin{proof}
  Suppose $f^{}_0$ as in Definition~\ref{dfi:LC-total-cover} is a minimally non-tight, optimal total-cover for $\cL^{}_0$. 
  Hence, there is a $y\in Y$ such that $|f^{}_0(y)|>1$. Let $\ell'_0\in f^{}_0(y)$ and define a new labeling $g$ where 
  $g(z)=f^{}_0(z)$ for all $z\in X\cup(Y\setminus\{y\})$ and $g(y)=f^{}_0(y)\setminus\{\ell'_0\}$. Note that $g(y)\neq\emptyset$
  and that every edge $(x,y)$ for $x\in N(y):=\{z\in X: (z,y)\in E\}$ is covered (for $f^{}_0$ is a total-cover). Moreover,
  clearly $\kappa(f^{}_0)=\kappa(g)$. Hence, $g$ is an optimal total-cover for $\cL^{}_0$ in which 
  $g(Y) = \sum_{y\in Y} |g(y)| < f^{}_0(Y)$, contradicting the minimality of $f^{}_0$. The result thus follows.
\end{proof}

\begin{notation}\label{nta:LC-sizes}
  For $\cL^{}_0$ being a \LC instance as in Definition~\ref{dfi:LC-inst}, define $r:=|X|$, $s:=|Y|$, 
  $m:=|E|$, $\lambda:=|L_0^{}|$, $\lambda':=|L_0'|$, $\pi_e:=|\Pi^0_e|$ for $e\in E$, and set 
  $\pi:=\sum_{e\in E}\pi_e$.
\end{notation}

\begin{problem}\label{pro:label-cover}
  For any $\rho>1$, a \LC instance $\cL^{}_0$ has covering promise $\rho$ if it falls in one of two cases:
  either there is a tight, optimal total-cover for $\cL^{}_0$ of cost $1$, or every tight, optimal total-cover for 
  $\cL^{}_0$ has cost at least $\rho$. The $\LC_\rho$ problem is a promise problem which receives a \LC instance 
  with covering promise $\rho$ (also known as a $\rho$-promise instance) as input and correctly classify it in one 
  of those two cases.
\end{problem}

Notice the behavior of $\LC_\rho$ is left unspecified for non-promise instances. Therefore, any answer is 
acceptable in such case. Due to this characteristic, the $\LC_\rho$ problem is also referred as a \emph{gap-problem}
with gap $\rho$ in the literature. 

The result below is the basis for the polynomial time hardness we shall exhibit.

\begin{theorem}[Dinur and Safra~\cite{DS04}]\label{thm:LC-inapprox}
  Let $c$ be any constant in $(0,1/2)$ and $\rho_c(s):=2^{(\log s)^{1-1/\delta_c(s)}}$ with $\delta_c(s):=(\log\log s)^{c}$.
  There are \LC instances $\cL^{}_0$ with covering-promise $\rho_c(s)$ such that it is \NP-hard to distinguish 
  between the cases in which $\kappa(\cL^{}_0)$ is equal to $1$ or at least $\rho_c(s)$.
\end{theorem}

The closer to $1/2$ the above constant $c$ gets, the larger the hardness of approximation factor becomes. 
Therefore, from now on we shall consider that $c$ is fixed to a value close to $1/2$. 

\begin{remark}\label{rem:DS-inst-size}
  Every \LC instance produced by Dinur and Safra's reduction is 
  feasible, has covering-promise $\rho_c(s)$, and satisfies the following relations:
  $r=s\floor{\delta_c(s)}$, $\lambda=\Theta(\rho_c(s))$, $\lambda'=\Theta(\rho_c(s)^{\delta_c(s)})=o(s)$, 
  $s\floor{\delta_c(s)}\leq m\leq s^2\floor{\delta_c(s)}$, and 
  $\pi\leq m\lambda\lambda'=O(s^2\delta_c(s)\rho_c(s)^{\delta_c(s)+1})=o(s^3)$, for $s$ as specified
  in Notation~\ref{nta:LC-sizes}. It is then immediate that each such instance has size $o(s^3)$.
\end{remark}

We now introduce a refined version of the \LC definitions, in which the vertices in the sets $X$ and $Y$ 
have their own copies of the label-sets $L_0^{}$ and $L_0'$, respectively. We then show that all structural and 
approximation properties are preserved in this new version.

\begin{definition}\label{dfi:ELC-inst}
  Let $\cL^{}_0=(G,L_0^{},L_0',\Pi^{}_0)$ be a feasible \LC instance and consider the sets $L_x:=\{(x,\ell_0^{}):\ell_0^{}\in L_0^{}\}$ 
  for each vertex $x\in X$, and $L'_y:=\{(y,\ell'_0):\ell'_0\in L_0'\}$ for each vertex $y\in Y$. Also, define the sets
  $L:=\bigcup_{x\in X} L_x$, $L':=\cup_{y\in Y} L'_y$, and $\Pi:=\bigcup_{(x,y)\in E}\Pi_{(x,y)}$,
  with
  \[
    \Pi_{(x,y)}:=\left\{\big((x,\ell_0^{}),(y,\ell'_0)\big):(\ell_0^{},\ell'_0)\in\Pi^0_{(x,y)}\right\}.
  \]
  The quadruple $\cL=(G,L,L',\Pi)$ is called a \emph{refinement} of $\cL^{}_0$.
\end{definition}

It is clear that $|L_x|=\lambda$ for each vertex $x\in X$, $|L_y|=\lambda'$ for each vertex $y\in Y$, $|\Pi_e|=\pi_e$ for 
each edge $e\in E$, $|\Pi|=\pi$, and that a labeling for $\cL$ is a mapping $f$ such that $x\mapsto f(x)\subseteq L_x$ for 
each vertex $x\in X$, and $y\mapsto f(y)\subseteq L'_y$, $f(y)\neq\emptyset$ for each vertex $y\in Y$. Furthermore, the 
remaining definitions and concepts can be adapted in a straight forward fashion, and the size of a refined instance 
is also $o(s^3)$.

\begin{lemma}\label{lem:LC-ELC}
  For any $\rho>0$, there is a one-to-one cost preserving correspondence between solutions to the $\LC_\rho$ 
  problem and to its refined version.
\end{lemma}
\begin{proof}
  It is easy to see that $f^{}_0$ is a (tight) total-cover for the $\LC_\rho$ problem if and only if $f$ is a
  (tight) total-cover for its refinement, where $f(x)=\{(x,\ell_0^{}):\ell_0^{}\in f^{}_0(x)\}$ for 
  every $x\in X$, and $f(y)=\{(y,\ell'_0)\in L'_y:\ell'_0\in f^{}_0(y)\}$ for every $y\in Y$ (or $f(y)=(y,f^{}_0(y))$ 
  if the total-covers are tight). Furthermore it is clear that $\kappa(f^{}_0)=\kappa(f)$, in any case. 
\end{proof}

Henceforth, all the \LC instances used are assumed to be of the refined kind.
For more information on the \LC problem and its applications, consult the survey by Arora and Lund~\cite{AL96}, 
the article by Moshkovitz and Raz~\cite{MR-10-JACM}, and the book by Arora and Barak~\cite{AB09}. 

\section{Reduction to pure Horn CNFs and a Polynomial Time Hardness Result}\label{sec:reduction}
Our first reduction starts with a \LC instance $\cL$ as input and produces a pure Horn CNF formula $\Phi$, which
defines a pure Horn function $h$. The driving idea behind this reduction is that of tying the cost of tight, 
optimal total-covers of $\cL$ to the size of clause minimum prime pure Horn CNF representations of $h$.

With this in mind, let $\cL=(G=(X,Y,E),L,L',\Pi)$ be a \LC instance (in compliance with 
Theorem~\ref{thm:LC-inapprox} and Definition~\ref{dfi:ELC-inst}), and let $d$ and $t$ be positive integers 
to be specified later. Both $d$ and $t$ will be used as (gap) amplification devices. For nonnegative integers 
$n$, define $[n]:=\{1,\ldots,n\}$.

Associate propositional variables $u(\ell)$ with every label $\ell\in L\cup L'$, $e(x,y,i)$ and $e(x,y,\ell',i)$ 
with every edge $(x,y)\in E$, every label $\ell'\in L'_y$ and every index $i\in[d]$. Let $v(j)$, for 
indices $j\in[t]$, be extra variables, and consider the following families of clauses:
\[
  \begin{array}{clr}
    \textrm{(a)} & \displaystyle u(\ell)\wedge u(\ell')\lra e(x,y,\ell',i)         
                 & \forall\ (x,y)\in E,\ (\ell,\ell')\in\Pi_{(x,y)},\ i\in[d]; \\[4mm]

    \textrm{(b)} & \displaystyle \bigwedge_{z\in N(y)} e(z,y,\ell',i)\lra e(x,y,i) 
                 & \forall\ (x,y)\in E,\ \ell'\in L'_y,\ i\in[d]; \\[5mm]

    \textrm{(c)} & \displaystyle e(x,y,i)\lra e(x,y,\ell',i) 
                 & \forall\ (x,y)\in E,\ \ell'\in L'_y,\ i\in[d]; \\[4mm]

    \textrm{(d)} & \displaystyle \bigwedge_{i\in[d]}\bigwedge_{(x,y)\in E}e(x,y,i)\lra u(\ell) 
                 & \forall\ \ell\in L\cup L'; \\[5mm]

    \textrm{(e)} & \displaystyle v(j)\lra u(\ell) 
                 & \forall\ j\in[t],\ \ell\in L\cup L'; \\[2mm]
  \end{array}
\]
where as before, $N(y):=\{x\in X:(x,y)\in E\}$ is the open neighborhood of the vertex $y\in Y$.

\begin{definition}\label{dfi:fcns}
  Let us call $\Psi$ and $\Phi$ the \emph{canonical} pure Horn CNF formulae defined, respectively,  by the families of 
  clauses (a) through (d) and by all the families of clauses above. Let $g$ and $h$ be, in that order, the pure Horn 
  functions they represent. 
\end{definition}

The construction presented above can be divided into two parts. The families of clauses appearing in $\Psi$, namely, 
clauses of type (a) through (d), form an independent core since the function $g$ is an exclusive component of the 
function $h$ (as we shall show). This core can be analysed and minimized separately from the remainder, and its role 
is to reproduce the structural properties of the \LC instance. In more details, clauses of type 
\begin{itemize}
\item[(a)]
  correspond to the constraints on the pairs of labels that can be assigned to each edge; 

\vspace{2mm}
\item[(b)]
  will assure that edges $(x,y)\in E$ are covered, enforcing the matching of the labels assigned to all the
  neighbors of vertex $y$;

\vspace{2mm}
\item[(c)]
  assure that if an edge can be covered in a certain way, then it can be covered in all legal ways --- thus implying 
  that it is not necessary to keep track of more than one covering possibility for each edge in clause minimum 
  prime representations;

\vspace{2mm}
\item[(d)]
  translate the total-cover requirement and reintroduce all the labels available ensuring that if a total-cover 
  is achievable, so are all the others; this reintroduction of labels is paramount to the proper functioning of 
  the reduction as explained below.
\end{itemize}

The family of clauses occurring in $\Phi\setminus\Psi$, namely, the clauses of type (e), constitutes the second
part of the construction. These clauses have the role of introducing an initial collection of labels, which sole 
purpose is to help achieve the claimed hardness of approximation result. The intended behavior is as follows.

Consider initially that $d=t=1$. It is known that given any subset of the variables of $h$ as input, the Forward 
Chaining procedure in any pure Horn CNF representation of $h$ will produce the same output (cf. 
Lemma~\ref{lem:horn-FC}). In particular, for the singleton $\{v(1)\}$, the output in $\Phi$ will be the set
with all the variables of $h$, and so will be the output obtained in any clause minimum prime pure Horn CNF
representing $h$.

The reintroduction of labels performed by the family of clauses (d) may allow for some clauses of type (e) to
be dropped without incurring in any loss. In slightly more details, as long as a subset of the family of 
clauses (e) introduces enough labels so that the Forward Chaining procedure in $\Phi$ is able to eventually 
trigger the family of clauses (d), the remaining clauses of type (e) can be dismissed. All the missing labels 
will be available by the end of the procedure's execution. It is not hard to see at this point that subsets 
of retained clauses of type (e) and total-covers of the \LC instance in which the reduction is based are in 
one-to-one correspondence.

Now, supposing that a clause minimum prime pure Horn CNF representation of $h$ resembles the canonical form $\Phi$, 
we just have to compensate for the number of clauses in $\Psi$ to obtain a distinguishable gap that mimics the 
one exhibited by the \LC (as a promise) problem. This is achieved by making the gap amplification parameter $t$ 
which the clauses of type (e) depend upon large enough. 

However, in principle, there is no guarantee that a clause minimum prime pure Horn CNF representation of $h$, say
$\Upsilon$, resembles $\Phi$ or that $\Upsilon$ has any clause of type (e) whatsoever. It may be advantageous 
to $\Upsilon$ to have prime implicates where $v(j)$, for $j\in[t]$, occurs in their subgoals or prime implicates 
with variables other than $u(\ell)$, for $\ell\in L\cup L'$, occurring as heads. Furthermore, the number of prime 
implicates in $\Upsilon$ involving $v(j)$ might simply not depend on the number of labels. Indeed, if $d=1$ as
we are supposing, whenever the number of edges $|E|$ turns out to be strictly smaller than $\kappa(\cL)$, the 
cost of an optimal total-cover for $\cL$, it would be advantageous for $\Upsilon$ to have prime implicates of the 
form $v(j)\lra e(x,y,1)$, for $(x,y)\in E$. This not only breaks the correspondence mentioned above, but it renders 
the gap amplification device $t$ innocuous and the whole construction useless.

We manage to overcome the above difficulty throughout a second amplification device, the parameter $d$ which 
the clauses of type (a) through (d) depend upon. As we shall prove in Lemma~\ref{lem:cnf-no-shortcuts}, setting 
$d=1+r\lambda+s\lambda'$ (which is strictly larger than the total number of labels available in $\cL$ --- cf. 
Notation~\ref{nta:LC-sizes} and Definition~\ref{dfi:ELC-inst}) allows us to control the shape of the prime 
implicates involving variables $v(j)$ in prime pure Horn clause minimum representations of $h$: they will be 
precisely some of the clauses of type (e). Moreover, after showing that the function $g$ is an exclusive 
component of the function $h$, we shall see that we do not need to concern ourselves with the actual form of 
clause minimum prime pure Horn CNF representations of $g$. Therefore, in a sense, the canonical form $\Phi$ has 
indeed a good resemblance to a clause minimum prime pure Horn CNF representing $h$, and the intended behavior is 
achieved in the end.

\subsection{Correctness of the CNF Reduction}
In this subsection, we formalize the discussion presented above. We will constantly use the canonical 
representations $\Phi$ and $\Psi$ to make inferences about the functions $h$ and $g$ they respectively
define, and such inferences will most of the time be made throughout Forward Chaining. We start with 
some basic facts about $\Phi$ and $\Psi$.

\begin{lemma}\label{lem:cnf-reduction-size}
  Let $d$ and $t$ be as above and let $r$, $s$, $m$, $\lambda$, $\lambda'$, and $\pi$ be as in 
  Notation~\ref{nta:LC-sizes}. It holds that the number of clauses and variables in $\Phi$ are, 
  respectively, 
  \[
    |\Phi|_c=(t+1)(r\lambda+s\lambda')+d(\pi+2m\lambda') 
    \quad\text{and}\quad 
    |\Phi|_v=t+dm(\lambda'+1)+r\lambda+s\lambda'.
  \]
  In $\Psi$, those numbers are, respectively, 
  \[
    |\Psi|_c=r\lambda+s\lambda'+d(\pi+2m\lambda')
    \quad\text{and}\quad
    |\Psi|_v=dm(\lambda'+1)+r\lambda+s\lambda'.
  \]
\end{lemma}
\begin{proof}
  For $\#(\alpha)$ denoting the number of clauses of type $(\alpha)$ in $\Phi$, simple counting arguments show that
  the equalities
  \[
    \begin{array}{lll}
      \#(\textsl{a})=d\pi            & \#(\textsl{c})=dm\lambda'         & \#(\textsl{e})=t(r\lambda+s\lambda') \\
      \#(\textsl{b})=dm\lambda'\quad & \#(\textsl{d})=r\lambda+s\lambda'\quad 
    \end{array}
  \]
  hold. For the number of variables, just notice there are $r\lambda+s\lambda'$ variables $u(\ell)$, $dm$ variables
  $e(x,y,i)$, $dm\lambda'$ variables $e(x,y,\ell',i)$, and $t$ variables $v(j)$. To conclude the proof, just 
  remember that the only difference between $\Phi$ and $\Psi$ is the absence of the family of clauses of type (e)
  in the latter.
\end{proof}

Considering the bounds for $r$, $m$, $\lambda$, $\lambda'$, and $\pi$ provided in Remark~\ref{rem:DS-inst-size},
the above result immediately implies that as long as the quantities $d$ and $t$ are polynomial in $s$, namely, 
the number of vertices in $Y$, the construction of $\Phi$ from $\cL$ can be carried out in polynomial time in
$s$. More meaningfully, it can be carried out in polynomial time in $n=|\Phi|_v$, the number of variables of 
$h$. 

We now establish the pure Horn function $g$ as an exclusive component of $h$. This structural result allows
us to handle $g$ in a somewhat black-box fashion. Specifically, as we shall see briefly, it is not required
of us to precisely know all the properties and details of a clause minimum representation of $g$. We can
mainly concentrate on the study of the prime implicates that might involve the variables in $h$ that are not 
in $g$, namely, the variables $v(j)$, for $j\in[d]$. 

\begin{lemma}\label{lem:g-excl-comp}
  The function $g$ is an exclusive component of the function $h$. Consequently, $g$ can be analysed and 
  minimized separately.
\end{lemma}
\begin{proof}
  Let $V_g$ be the set of variables occurring in $\Psi$. By definition, these are the variables the function 
  $g$ depends upon. Since no clause in $\Psi$ has head outside $V_g$, it is immediate that $V_g$ is closed 
  under Forward Chaining in $\Phi$. As $\Phi$ represents $h$, Lemma~\ref{lem:FC-exclusive-set} then implies 
  that 
  $
    \cX(V_g) := \{C\in\cI(h): \mathsf{Vars}(C)\subseteq V_g\}
  $ 
  is an exclusive family for $h$. This gives that $\Psi=\Phi\cap\cX(V_g)$ is an $\cX(V_g)$-exclusive component 
  of $h$ (cf. Definition~\ref{dfi:exclusive-comp}) and therefore, that $g$ is an exclusive component of $h$.
  Now, using Lemma~\ref{lem:subfunc-switch}, we obtain a proof of the second claim as wished.
\end{proof}

With some effort, it is possible to prove that $\Psi$ is a clause minimum prime pure Horn CNF representation of 
$g$. For our proofs however, a weaker result suffices.

\begin{lemma}\label{lem:cnf-psi-lb}
  Let $\Theta$ be a clause minimum prime pure Horn CNF representation of $g$. We have
  $|\Psi|_c/(\lambda+\lambda')\leq |\Theta|_c\leq |\Psi|_c$.
\end{lemma}
\begin{proof}
  The upper bound is by construction. For the lower bound, observe that each variable of $\Psi$ appears no more 
  than $\lambda+\lambda'$ times as a head. As in any clause minimum representation of $g$ they must appear as 
  head at least once, the claim follows.
\end{proof}

The next lemma is a useful tool in showing whether two different representations of the pure Horn function $h$
are equivalent.

\begin{lemma}\label{lem:FC(v(j))}
  For all indices $j\in[t]$, it holds that $F_{h}(\{v(j)\})=V_g\cup\{v(j)\}$.
\end{lemma}
\begin{proof}
  It is enough to show that $\{e(x,y,i):(x,y)\in E,i\in[d]\}\subseteq F_{\Phi}(\{v(j)\})$, for a fixed $j\in[t]$. 
  The inclusion would be false if there existed a label $\ell''\in L\cup L'$ such that $u(\ell'')\not\in F_{\Phi}(\{v(j)\})$.
  As for every label $\ell\in L\cup L'$, $v(j)\lra u(\ell)$ is a clause in $\Phi$, this cannot happen. Hence, the 
  inclusion holds, implying the claim.
\end{proof}

The next couple of lemmas deal with the structure of prime implicates involving the variables $v(j)$, for 
$j\in[t]$. The first states a simple, but useful fact which is valid in any representation of $h$. The second 
proves how the amplification device depending on the parameter $d$ shapes those prime implicates in clause 
minimum representations of $h$ to the desired form: $v(j)\lra u(\ell)$, with $\ell\in L\cup L'$.
 
\begin{lemma}\label{lem:v(j)-quadratic}
  A variable $v(j)$, for some index $j\in[t]$, is never the head of an implicate of $h$. Moreover, every prime 
  implicate of $h$ involving $v(j)$ is quadratic.
\end{lemma}
\begin{proof}
  The first claim is straight forward as all implicates of $h$ can be derived from $\Phi$ by resolution,
  and $v(j)$ is not the head of any clause of $\Phi$. By Lemma~\ref{lem:FC(v(j))}, $v(j)\lra z$ is an 
  implicate of $h$ for all $z\in V_g$. Since $h$ is a pure Horn function, the claim follows.
\end{proof}

\begin{lemma}\label{lem:cnf-no-shortcuts}
  Let $d=1+r\lambda+s\lambda'$. In any clause minimum prime pure Horn CNF representation of $h$, the prime implicates 
  involving the variables $v(j)$ have the form $v(j)\lra u(\ell)$, for all indices $j\in[t]$, and for some labels 
  $\ell\in L\cup L'$.	
\end{lemma}
\begin{proof}
  Let $\Upsilon=\Theta\wedge\Gamma$ be a clause minimum prime pure Horn CNF representation of $h$, with $\Theta$ being 
  a clause minimum pure Horn CNF representation of $g$. According to Lemma~\ref{lem:v(j)-quadratic}, all prime 
  implicates of $h$ involving the variables $v(j)$ are quadratic. So, for all indices $j\in[t]$ and all indices 
  $i\in[d]$ define the sets
  \begin{align*}
    \Gamma^j_0 & := \Gamma\cap\{v(j)\lra u(\ell): \ell\in L\cup L'\},\\
    \Gamma^j_i & := \Gamma\cap\{v(j)\lra e(x,y,i),\,
                                v(j)\lra e(x,y,\ell',i): (x,y)\in E, \ell'\in L'_y\}.
  \end{align*}

  Our goal is to show that the chosen value for the parameter $d$ forces all the sets $\Gamma^j_i$ to be
  simultaneously empty and consequently, that all the prime implicates involving the variables $v(j)$
  in clause minimum pure Horn CNF representations of $h$ have the claimed form. We shall accomplish this 
  in two steps. 

  Let $j\in[t]$. We first show that if a set $\Gamma^j_i\neq\emptyset$ for some index $i\in[d]$, then 
  $\Gamma^j_i\neq\emptyset$ for all indices $i\in[d]$, simultaneously. 

  \begin{claim}
    All clauses of type (d) have the same body. Therefore, during the execution of the Forward Chaining procedure
    from $\{v(j)\}$, either they all trigger simultaneously or none of them do. The reason for them not to trigger 
    is the absence of some variable $e(x,y,i)$, with $(x,y)\in E$ and $i\in[d]$, in the Forward Chaining closure 
    from $\{v(j)\}$, i.e, $e(x,y,i)\not\in F_{\Upsilon}(\{v(j)\})$.
  \end{claim}
  \begin{proof}
    A simple inspection of the families of clauses shows that Claim (1) holds. 
  \end{proof}

  Now, for each index $i\in[d]$, let $\Xi^{}_i$ be the collection of clauses of types (a), (b), and (c) that 
  depend on $i$.

  \begin{claim}
    It holds that
    \[
      e(x,y,i) \in F_{\Gamma^j_0\cup\Xi^{}_i} (\{v(j)\})
      \quad\text{if and only if}\quad
      e(x,y,i') \in F_{\Gamma^j_0\cup\Xi^{}_{i'}} (\{v(j)\}),
    \]
    for all indices $i,i'\in[d]$, with $i\neq i'$.
  \end{claim}
  \begin{proof}
    Notice that the families of clauses (a), (b), and (c) are completely symmetric with respect to the indexing 
    variable $i$. Moreover, for $i_1\neq i_2$, the clauses indexed by $i_1$ do not interfere with the clauses 
    indexed by $i_2$ during an execution of the Forward Chaining procedure. In other words, variables depending 
    upon $i_1$ do not trigger clauses indexed by $i_2$, and vice-versa. These two properties, symmetry and non 
    interference, proves Claim (2). 
  \end{proof}

  \begin{claim}
    If there is a variable $e(x,y,i)$, with $(x,y)\in E$ and $i\in[d]$, such that
    \[
      e(x,y,i)\not\in F_{\Gamma^j_0\cup\left(\bigcup_{i\in[d]}\Xi_i^{}\right)}(\{v(j)\})
    \]
    then
    \[
      e(x,y,i)\not\in F_{\Gamma^j_0\cup\left(\bigcup_{i\in[d]}\Xi_i^{}\right)\cup\left(\bigcup_{i'\neq i}\Gamma^j_{i'}\right)}
      (\{v(j)\}).
    \]
    Moreover, this implies that $\Gamma^j_i\neq\emptyset$.
  \end{claim}
  \begin{proof}
    The symmetry and non interference properties of families of clauses (a), (b), and (c) also justifies the first
    part of Claim (3). To see it, just notice that were the claim to be false, the prime implicates in $\Gamma^j_{i'}$ 
    would be trigging clauses involving the variable $e(x,y,i)$ in an execution of the Forward Chaining procedure. 
    Since $i'\neq i$, this cannot happen. The second part follows immediately from the validity of the first part
    together with the fact that $\Upsilon$ represents $h$. 
  \end{proof}

  To finish the first step, notice that since Claim (3) is valid for any $i\in[d]$, Claim (2) implies that if 
  $\Gamma^j_i\neq\emptyset$ for some index $i\in[d]$, then $\Gamma^j_i\neq\emptyset$ for all indices $i\in[d]$, 
  simultaneously.

  For the second step, suppose that $\Gamma^j_i\neq\emptyset$ for all indices $i\in[d]$. We then have that
  \[
    \gamma:=\sum_{i\in[d]}|\Gamma^j_i|\geq d=1+r\lambda+s\lambda'=1+|L\cup L'|, 
  \]
  that is, $\gamma$ is strictly larger than the number of all available labels in $\cL$. This implies that 
  the following pure Horn CNF
  \begin{align*}
    \Delta_j & := \left(\Upsilon\setminus\bigcup_{i\in d}\Gamma^j_i\right)\cup\Big\{v(j)\lra u(\ell):\ell\in L\cup L'\Big\}\\
             &\ = \Theta\cup\left(\left(\Gamma\setminus\bigcup_{i\in d}\Gamma^j_i\right)\cup
                  \Big\{v(j)\lra u(\ell):\ell\in L\cup L'\Big\}\right),
  \end{align*}
  has fewer clauses than $\Upsilon$ (or more precisely, it implies that $|\Delta_j|_c <= |\Upsilon|_c - 1$). 

  Now, since that $\Theta$ is a (clause minimum) representation of the exclusive component $g$, and that the
  set of clauses $\{v(j)\lra u(\ell):\ell\in L\cup L'\}$ makes all available labels reachable by Forward
  Chaining from $\{v(j)\}$, it follows that $F_{\Delta_j}(\{v(j)\})=V_g\cup\{v(j)\}$. Furthermore, the change
  in clauses did not influence the Forward Chaining procedure from any other variable (other than $v(j)$), 
  and thus $F_{\Delta_j}(\{w\})=F_{\Upsilon}(\{w\})$ for all variables $w\neq v(j)$. Thus, Lemma~\ref{lem:FC(v(j))} 
  implies that $\Delta_j$ is a representation of $h$.

  We then have that $\Delta_j$ is a shorter representation for $h$, contradicting the optimality of $\Upsilon$.
  Therefore, the sets $\Gamma^j_i=\emptyset$ for all indices $i\in[d]$. As the above arguments do not depend on 
  any particular value of $j$, they can be repeated for all of them.
\end{proof}

The next property we can show more generally for any prime and irredundant CNF of $h$. 

\begin{definition}\label{dfi:cnf-tcover}
  Let $d=1+r\lambda+s\lambda'$ and let $\Upsilon$ be a prime and 
  irredundant pure Horn CNF representation of $h$. For each $j\in[t]$, consider the set 
  \[
    S_j=\left\{\ell\in L\cup L':v(j)\lra u(\ell) \in \Upsilon\right\}
  \] 
  and define the function  $f_j:X\ra L, Y\ra L'$ given by $f_j(x)=S_j\cap L_x$ for vertices $x\in X$ and 
  $f_j(y)=S_j\cap L'_y$ for vertices $y\in Y$.
\end{definition}

The next three lemmas provide important properties of the functions $f_j$ above.

\begin{lemma}\label{lem:cnf-tightness-up}
  Let $\Upsilon$ be as in the above Definition. For all indices $j\in[t]$ and vertices $y\in Y$, 
  it holds that $|f_j(y)|\leq 1$.
\end{lemma}
\begin{proof}
  Let $\Upsilon$ be as in Definition~\ref{dfi:cnf-tcover} and suppose indirectly that the claim is false, 
  that is, there is an index $j\in[t]$ and a vertex $y\in Y$ such that $|f_j(y)| > 1$.

  During the proof, recall that the chosen value for the parameter $d$ implies, according to 
  Lemma~\ref{lem:cnf-no-shortcuts}, that all prime implicates of $\Upsilon$ involving the variable 
  $v(j)$ must have the form $v(j)\lra u(\ell)$, with $\ell\in L\cup L'$.

  Let $\ell'\in f_j(y)$ and define the expression 
  \[
    \Upsilon':=\Upsilon\setminus\{v(j)\lra u(\ell')\}.
  \]

  It is enough to show that $F_{\Upsilon'}(\{v(j)\})=V_g\cup\{v(j)\}$, that is, that $\Upsilon'$ is also a 
  representation of $h$ (cf. Lemma~\ref{lem:FC(v(j))}). Suppose that is not the case. Since $\Upsilon$ and
  $\Upsilon'$ differ only in the clause $v(j)\lra u(\ell')$, it must be the case that 
  $u(\ell')\not\in F_{\Upsilon'}(\{v(j)\})$. This happens if the clause of type (d) associated to $u(\ell')$ 
  is not trigged. For this to occur, there must be an edge $(x,y)\in E$ and an index $i\in[d]$ such 
  that $e(x,y,i)\not\in F_{\Upsilon'}(\{v(j)\})$. 

  Now, for $y$ and $i$ as above, notice that: 
  (i) variable $e(x,y,i)$ would be included in $F_{\Upsilon'}(\{v(j)\})$ as long as there were a label in 
  $L'_y$ such that the corresponding clause of type (b) were trigged; and
  (ii) once such clause of type (b) were trigged, the appropriated clauses of type (c) would trigger, thus
  making the other clauses of type (b) associated to $y$ and $i$ to also trigger. 

  Therefore, for $e(x,y,i)$ to not belong to $F_{\Upsilon'}(\{v(j)\})$, it must be the case that for every 
  label $\ell''\in f_j(y)\setminus\{\ell'\}$ there exists a vertex $z(\ell'')\in N(y)$ for which
  \[
    e(z(\ell''),y,\ell'',i)\not\in F_{\Upsilon'}(\{v(j)\}).
  \]
  
  For this latter relation to be true, we must have that the clauses 
  \begin{equation}\label{eq:trigged}
    u(\ell)\wedge u(\ell'')\lra e(z(\ell''),y,\ell'',i)
  \end{equation}
  are not trigged in the Forward Chaining procedure on $\Upsilon'$ starting with $\{v(j)\}$, for every 
  label $\ell\in f_j(z(\ell''))$ with $(\ell,\ell'')\in\Pi_{(z(\ell''),y)}$. 

  However, according to Definition~\ref{dfi:cnf-tcover}, for each label $\ell''\in f_j(y)\setminus\{\ell'\}$
  and each label $\ell\in f_j(z(\ell''))$, there are clauses $v(j)\lra u(\ell'')$ and $v(j)\lra u(\ell)$,
  respectively, in $\Upsilon$ and, consequently, in $\Upsilon'$. This implies that the clauses \eref{eq:trigged} 
  are trigged, which implies that $u(\ell')\in F_{\Upsilon'}(\{v(j)\})$, which then implies that $\Upsilon'$ is 
  also a representation of $h$. Since this contradicts the irredundancy of $\Upsilon$, it follows that 
  $|f_j(y)|\leq 1$, thus concluding the proof.
\end{proof}

\begin{lemma}\label{lem:cnf-tightness-down}
  Let $\Upsilon$ be a clause minimum prime pure Horn CNF of $h$. Then it is prime and irredundant, 
  so Definition~\ref{dfi:cnf-tcover} applies. We claim that for all indices $j\in[t]$ and vertices $y\in Y$, 
  it holds that $|f_j(y)| \geq 1$.
\end{lemma}
\begin{proof}
  Suppose that the claim is false, that is, 
  there is an index $j\in[t]$ and a vertex $y\in Y$ such that $|f_j(y)| = 0$.

  Then clauses $v(j)\lra u(\ell')$, for all labels $\ell'\in L'_y$, are absent from 
  $\Upsilon$. Recall that the chosen value for the parameter $d$ implies that all prime implicates 
  of $\Upsilon$ involving $v(j)$ are quadratic (Lemma~\ref{lem:cnf-no-shortcuts}).

  Thus, no clause of type (a) dependent on the vertex $y$ is trigged during a Forward 
  Chaining from $\{v(j)\}$ and hence, no clauses of type (b) and of type (c) dependent on $y$ are 
  trigged either. This gives that the variables $e(x,y,i)$, for all vertices $x\in N(y)$ and all
  indices $i\in[d]$, do not belong to the Forward Chaining closure (from $\{v(j)\}$). Therefore,
  no clause of type (d) is trigged and no label $\ell'\in L'_y$ is reintroduced. In other words, 
  it is the case that $u(\ell')\not\in F_{\Upsilon}(\{v(j)\})$ and hence, that 
  $F_{\Upsilon}(\{v(j)\}) \neq F_h(\{v(j)\})$. By Lemma~\ref{lem:FC(v(j))}, $\Upsilon$ does not
  represent $h$, a contradiction. So, it must be the case that $|f_j(y)|\geq 1$.
\end{proof}

Combining the two lemmas above, we have the following tight result.
\begin{corollary}\label{cor:cnf-tightness}
  Let $\Upsilon$ be a clause minimum prime pure Horn CNF of $h$. Then, for all indices $j\in[t]$ 
  and vertices $y\in Y$, it holds that $|f_j(y)| = 1$.
  \qed
\end{corollary}

The next lemma shows that the functions $f_j$ are indeed tight total-covers.

\begin{lemma}\label{lem:cnf-tcover}
  Let $\Upsilon$ be a clause minimum prime pure Horn CNF of $h$. For each index $j\in[t]$, the function $f_j$ 
  is a tight total-cover for $\cL$.
\end{lemma}
\begin{proof}
  Let $j\in[t]$. By construction (cf. Definition~\ref{dfi:cnf-tcover}), $f_j$ is a labeling for $\cL$.
  Suppose however, that $f_j$ is not a total-cover. Hence, there exists an edge $(x,y)\in E$ 
  and a label $\ell'\in f_j(y)$ such that for all labels $\ell\in f_j(x)$, it holds that 
  $(\ell,\ell')\not\in\Pi_{(x,y)}$. Since $\Upsilon$ is clause 
  minimum, no variable $e(x,y,\ell',i)$ belongs to $F_{\Upsilon}(\{v(j)\})$, for any index $i\in[d]$.
  This is so because the variables $e(x,y,\ell',i)$ do not occur as heads in any clause of $\Upsilon$
  (cf. Lemma~\ref{lem:cnf-no-shortcuts}). Now, using Lemma~\ref{lem:FC(v(j))}, we obtain that $\Upsilon$ 
  does not represent $h$, a contradiction. To conclude the proof, just notice that Corollary~\ref{cor:cnf-tightness} 
  implies that $f_j$ is tight.
\end{proof}

In order to relate the size of a clause minimum representation of $h$ to the cost of an optimal solution 
to $\cL$, we need a comparison object. Let $f$ be a tight total-cover for $\cL$ and consider the following 
subfamily of clauses:
\[
  \begin{array}{clr}
    \textrm{(e')} & \displaystyle v(j)\lra u(\ell) \quad        
                  & \forall\ j\in[t],\ x\in X,\ y\in Y,\ \ell\in f(x)\cup f(y), \\[2mm]
  \end{array}
\]
with $f(x)\subseteq L_x$ and $f(y)\subseteq L'_y$. Let $\Phi_f$ be the \emph{refined canonical} 
(with respect to $f$) pure Horn CNF formula resulting from the conjunction of $\Psi$ with the 
clauses of type (e').

\begin{lemma}\label{lem:cnf-represents}
  $\Phi_f$ represents $h$.
\end{lemma}
\begin{proof}
  Suppose the opposite. As $g$ is also an exclusive component of $\Phi_f$, that means 
  $u(\ell'')\not\in F_{\Phi_f}(\{v(j)\})$ for some label $\ell''\in L\cup L'$ and index $j\in[t]$. For that to 
  happen, for any index $i\in[d]$ there must be an edge $(x,y)\in E$ such that $e(x,y,i)\not\in F_{\Phi_f}(\{v(j)\})$ 
  and for every $\ell'\in L'_y$ there is a vertex $z\in N(y)$ such that $e(z,y,\ell',i)\not\in F_{\Phi_f}(\{v(j)\})$
  as well. But since $f$ is a tight total-cover, there is a pair of labels $(\ell_z,\ell'_y)\in\Pi_{(z,y)}$ triggering
  the clause $u(\ell_z)\wedge u(\ell'_y)\lra e(z,y,\ell'_y,i)$ as both $u(\ell_z)$ and $u(\ell_y')$ belong to
  $F_{\Phi_f}(\{v(j)\})$, contradicting $e(z,y,\ell'_y,i)\not\in F_{\Phi_f}(\{v(j)\})$. Therefore, 
  $F_{\Phi_f}(\{v(j)\})=V_g\cup\{v(j)\}$ and the claim follows by Lemma~\ref{lem:FC(v(j))}.
\end{proof}

\begin{lemma}\label{lem:cnf-optimality}
  Let $\Upsilon$ be a clause minimum prime pure Horn CNF of $h$, and let us define $f_j$, for all indices 
  $j\in[t]$, as in Definition~\ref{dfi:cnf-tcover}. It holds that each $f_j$ is a tight, minimum cost total-cover 
  for $\cL$.
\end{lemma}
\begin{proof}
  Let $\Upsilon=\Theta\wedge\Gamma$ where $\Theta$ is an optimal representation of $g$ and $\Gamma$ consists of
  clauses of type (e) (cf. Lemma~\ref{lem:cnf-no-shortcuts}).
  As assured by Lemma~\ref{lem:cnf-tcover}, $f_j$ are tight total-covers for $\cL$, for each and every
  index $j\in[t]$. Notice that since $\Upsilon$ is clause minimum, it follows that all these tight
  total-covers have the same cost, i.e., $\kappa(f_{j})=\kappa(f_{k})$ for all $j,k\in[t]$.

  We then have that
  \begin{equation}\label{eq:sizemin}
    \sum_{j\in[t]} (\kappa(f_j)r+s) = t(\kappa(f_j)r+s) = |\Gamma|_c \leq |\Phi_{f}\setminus\Psi|_c = t(\kappa(f)r+s),
  \end{equation}
  where $f$ is any tight total-cover for $\cL$ and $\Phi_f$ is the refined canonical (w.r.t. $f$) formula 
  as in Lemma~\ref{lem:cnf-represents}. In particular, Equation~\eref{eq:sizemin} holds even when $f$ is
  a tight, minimum cost total-cover, thus implying that $f_j$ is optimal as claimed.
\end{proof}

\begin{remark}\label{rem:diff-tcovers}
  The tight, optimal total-covers $f_j$ and $f_k$, for $j,k\in[t]$ and $j\neq k$, might be different. As they have
  the same optimal cost, any one of them can be exhibited as solution to $\cL$.
\end{remark}

The following corollary summarizes the work done so far.

\begin{corollary}\label{cor:upsilon-ub}
  Let $\Psi$ be as in Definitions~\ref{dfi:fcns} and $\Upsilon$ be a clause minimum prime pure Horn CNF of $h$.
  It holds that 
  \[
    |\Psi|_c/(\lambda+\lambda')\leq |\Upsilon|_c - t(\kappa(f)r + s)\leq~|\Psi|_c,
  \] 
  where $\kappa(f)r + s$ is the total number of labels in a tight optimal total-cover $f$ for $\cL$.
  \qed
\end{corollary}

\subsection{The CNF Hardness Result}
We are now able to prove the main result of this section.

\begin{theorem}\label{thm:cnf-hardness}
  Let $c$ be a fixed constant close to $1/2$. Unless $\PT=\NP$, the minimum number of clauses of a pure Horn 
  function on $n$ variables cannot be approximated in polynomial time (depending on $n$) to within a factor of 
  \[ 
    \rho_c(n^{\varepsilon})\geq 2^{\varepsilon(\log n)^{1-1/\delta_c(n)}}=2^{\log^{1-o(1)} n},
  \]
  where $\delta_c(n)=(\log\log n)^{c}$, even when the input is restricted to CNFs with $O(n^{1+2\varepsilon})$ 
  clauses, for some $\varepsilon\in(0,1/4]$.
\end{theorem}
\begin{proof}
  Let $\cL_0$ be a Dinur and Safra's \LC promise instance (cf. Theorem~\ref{thm:LC-inapprox}) and let $\cL$ be 
  its equivalent refined version (cf. Definition~\ref{dfi:ELC-inst}). 

  Let $\Phi$ be the canonical formula constructed from $\cL$ with the parameter $d$ set to $d=1+r\lambda+s\lambda'$ 
  (cf. Lemma~\ref{lem:cnf-no-shortcuts}), and let $h$ be the pure Horn function defined by $\Phi$. Let $\Upsilon$ be 
  a pure Horn CNF representation of $h$ obtained by some exact clause minimization algorithm when $\Phi$ is given as 
  input.

  Recall Notation~\ref{nta:LC-sizes} and for convenience, let $\delta=\delta_c(s)$ and $\rho=\rho_c(s)$. 
  Substituting the quantities established in Lemma~\ref{lem:cnf-reduction-size} into the bounds given by
  Corollary~\ref{cor:upsilon-ub}, and applying the values given in Remark~\ref{rem:DS-inst-size}, we have
  that
  \begin{align*}
    |\Upsilon|_c & \geq t(\kappa(f)r + s) + \frac{d(\pi+2m\lambda') + r\lambda+s\lambda'}{\lambda+\lambda'} \\
                 & \geq st(\kappa(f)(\delta-1) + 1) + \Omega(s^2\delta\rho^{\delta}) \\ 
                 & \geq st(\kappa(f)(\delta-1) + 1) + \omega(s^{2}),
  \end{align*}
  and
  \begin{align*}
    |\Upsilon|_c & \leq t(\kappa(f)r + s) + d(\pi+2m\lambda') + r\lambda+s\lambda' \\
                 & \leq st(\kappa(f)\delta + 1) + O(s^3\delta\rho^{2\delta+1}) \\ 
                 & \leq st(\kappa(f)\delta + 1) + o(s^{4}).
  \end{align*}

  Similarly, for the number of variables we have that
  \begin{align*}
    t \leq |\Upsilon|_v & = t + dm(\lambda'+1) + r\lambda+s\lambda' \\
                        & \leq t + O(s^3\delta\rho^{2\delta}) \\
                        & \leq t + o(s^{4}).
  \end{align*}

  Now, choosing $\varepsilon>0$ such that $t=s^{1/\varepsilon}=\Omega(s^4)$ and supposing that $s\lra\infty$,
  we obtain the asymptotic expressions
  \[
    |\Upsilon|_c = s^{(1+1/\varepsilon)}\delta_c(s)\kappa(f)(1+o(1))
    \qquad\text{and}\qquad
    |\Upsilon|_v = s^{1/\varepsilon}(1+o(1)).
  \]
	
  Bringing the existing gaps of the \LC promise instances into play, we then obtain the following dichotomy
  \begin{align*}
    \kappa(\cL)=1            & \Longrightarrow |\Upsilon|_c\leq s^{(1+1/\varepsilon)}\delta_c(s)(1+o(1)),\\
    \kappa(\cL)\geq\rho_c(s) & \Longrightarrow |\Upsilon|_c\geq s^{(1+1/\varepsilon)}\delta_c(s)\rho_c(s)(1+o(1)).
  \end{align*}

  Letting $n=|\Upsilon|_v$ and relating $|\Upsilon|_c$ to the number of variables of $h$, 
  the above dichotomy reads as
  \begin{align*}
    \kappa(\cL)=1            & \Longrightarrow |\Upsilon|_c\leq n^{(1+\varepsilon)}\delta_c(n^{\varepsilon})(1+o(1)),\\
    \kappa(\cL)\geq\rho_c(s) & \Longrightarrow |\Upsilon|_c\geq n^{(1+\varepsilon)}\delta_c(n^{\varepsilon})\rho_c(n^{\varepsilon})(1+o(1)),
  \end{align*}
  giving a hardness of approximation factor of $\rho_c(n^{\varepsilon})$ for the pure Horn CNF minimization problem
  (cf. Theorem~\ref{thm:LC-inapprox}). That is, any polynomial time algorithm that approximates the number of 
  clauses of a pure Horn CNF representation of $h$ to within a factor better than $\rho_c(n^{\varepsilon})$ can
  be used to solve the \LC promise problem for $\cL$. This in turn would show that $\PT=\NP$.
	
  To conclude the proof, notice that since $\log\varepsilon < 0$ and $n\lra\infty$, the gap
  \begin{align*}
    \rho_c(n^{\varepsilon}) 
    &    = 2^{(\varepsilon\log n)^{1-1/\delta_c(n^\varepsilon)}}\\
    & \geq 2^{\varepsilon(\log n)^{1-1/\delta_c(n)}}\\
    &    = 2^{(\log n)^{1-1/\delta_c(n)+\log\varepsilon/\log\log n}}\\
    &    = 2^{(\log n)^{1-o(1)}},
  \end{align*}
  where the $\log\varepsilon/\log\log n$ in the exponent is negligible compared to $1/\delta_c(n)$
  as $\delta_c(n)=o(\log\log n)$, and also notice that the number of clauses 
  \[
  n^{(1+\varepsilon)}\delta_c(n^{\varepsilon})\leq n^{1+2\varepsilon}.
  \]
\end{proof}

\section{Pure Horn 3-CNFs and Minimizing the Number of Literals}\label{sec:3-cnf}
In this section, we extend our hardness result in two ways. First, we prove that it still holds when
the pure Horn function $h$ is represented through a pure Horn 3-CNF, that is, in which each clause 
has at most three literals (and at least two, since $h$ is pure Horn). Second, we build upon this
first extension and prove that a similar bound holds when trying to determine a literal minimum pure 
Horn 3-CNF representation of $h$. 

Let again $d=1+r\lambda+s\lambda'$ and $t$ be a positive integer to be specified later. As in the 
previous section, both $d$ and $t$ will be used as (gap) amplification devices.

A brief inspection of our construction in Section~\ref{sec:reduction} shows that the clauses of type 
(b) and (d) may have arbitrarily long subgoals, with long meaning strictly more than three 
literals. The idea is then to modify the construction locally so that each long clause is replaced
by a gadget consisting of a collection of quadratic or cubic new clauses. Each gadget is designed to 
preserve the original logic implications of the clause it replaces.
  
Specifically, we replace the clauses of type (b) in a similar way to what is done in the reduction from 
SAT to 3-SAT (cf. Garey and Johnson~\cite{GJ79}), that is, in a \emph{linked-list} fashion. For each 
vertex $y\in Y$, let $\dg(y):=|N(y)|$ be its degree and let $\angleb{z_y^1,z_y^2,\ldots,z_y^{\dg(y)}}$ be 
an arbitrary, but fixed ordering of its neighbors. Associate new propositional variables $e(\beta,x,y,\ell',i)$ 
with all edges $(x,y)\in E$, all labels $\ell'\in L'_y$, all indices $i\in[d]$, and all indices
$\beta\in[\dg(y)-2]$. As before, we have that $d=1+r\lambda+s\lambda'$. Replace the clauses of type (b) 
by the families of clauses below:

\begin{description}
  \item[(b$_1$)] \raggedright $\displaystyle \bigwedge_{z\in N(y)}e(z,y,\ell',i)\lra e(x,y,i)$\hfill
                 $\forall\ (x,y)\in E$, $\ell'\in L'_y$, $i\in[d]$, $\dg(y)\leq 2;$ \\[2mm]
	
  \item[(b$_2$)] \raggedright $\displaystyle e(z_y^1,y,\ell',i)\wedge e(z_y^2,y,\ell',i)\lra e(1,x,y,\ell',i)$\\[.5mm]
                 \raggedleft $\forall\ (x,y)\in E$, $\ell'\in L'_y$, $i\in[d]$, $\dg(y)\geq 3;$ \\[2mm]
	
  \item[(b$_3$)] \raggedright $\displaystyle e(z_y^{\beta+2},y,\ell',i)\wedge e(\beta,x,y,\ell',i)\lra e(\beta+1,x,y,\ell',i)$\\[.5mm]
                 \raggedleft $\forall\ (x,y)\in E$, $\ell'\in L'_y$, $i\in[d]$, $\beta\in[\dg(y)-3];$ \\[2mm]
	
  \item[(b$_4$)] \raggedright $\displaystyle e(z_y^{\dg(y)},y,\ell',i)\wedge e(\dg(y)-2,x,y,\ell',i)\lra e(x,y,i)$\\[.5mm]
                 \raggedleft $\forall\ (x,y)\in E$, $\ell'\in L'_y$, $i\in[d]$, $\dg(y)\geq 3.$ \\[2mm]
\end{description}

The clauses of type (b$_1$) are exactly the quadratic and cubic clauses of type (b) in the original 
construction. The clauses of types (b$_2$), (b$_3$), and (b$_4$) rely on the new variables $e(\beta,x,y,\ell',i)$ 
to handle the remaining original clauses of type (b) through a series of split and link operations. 
It is not hard to see that these new families of clauses retain the symmetry and non interference 
properties possessed by the original ones they replaced. Furthermore, they will be part of an exclusive
component. These characteristics, as we shall see, contributes to the transference of most of the lemmas 
from the previous section in a rather verbatim fashion.

Trying to apply the same technique to the clauses of type (d) generates the following problem. 
Recall that $m=|E|$.
All the $r\lambda+s\lambda'$ clauses of type (d) have the same long subgoal with $dm$ literals; 
they only differ in their heads. In a first attempt to emulate what was done above to the clauses 
of type (b), we may decide to replace this long subgoal with a single linked list and to replicate
its last node once for each label $\ell\in L\cup L'$. After ordering these literals in an arbitrary, 
but fixed order, introducing $dm-2$ new variables, say $e(\zeta)$ for $\zeta\in[dm-2]$, and performing 
split and link operations, we obtain the linked list whose last node has subgoal 
\mbox{$e(x',y',i')\wedge e(dm-2)$}, for some edge $(x',y')\in E$ and some index $i'\in[d]$. This subgoal
is then replicated, thus spanning the clauses 
\[
  e(x',y',i')\wedge e(dm-2)\lra u(\ell),
\]
for all labels $\ell\in L\cup L'$. Now, it is not hard to see that the prime implicates 
\mbox{$v(1)\lra e(x',y',i')$} and \mbox{$v(1)\lra e(dm-2)$} completely bypass the amplification 
device dependent on the parameter $d$, reintroducing all the available labels. Similar to what was 
said before, this renders the gap amplification device dependent on parameter $t$ innocuous and the 
whole construction useless. It is also not hard to see that introducing a new amplification device 
and recurring on the whole idea does not solve the problem as we end up in a similar situation.

In a second attempt to emulate what was done to the clauses of type (b), we may decide to introduce 
different linked lists to different clauses of type (d). Notice that this implies 
that each of these lists must be indexed by one of the available labels in $L\cup L'$. In order for 
this new collection of clauses to be reached in an execution of the Forward Chaining procedure, this 
indexing scheme must be back propagated into the clauses of types (a), (b$_1$) through (b$_4$), and (c). 
That is, every clause of the exclusive component must now be indexed by a label in $L\cup L'$. With
some thought, it is possible to realize the following. First, this label indexing of clauses is 
performing a job similar to the one played by parameter $d$, in the sense that the latter could in 
principle be dropped --- and replacing one scheme by the other leaves the number of clauses in the 
exclusive component in the same order of magnitude. Second, this operation significantly changes the
meaning of the clauses being used as the labels would then be reintroduced in a somewhat independent
fashion. While the high degree of symmetry guarantees that all labels would be reintroduced, the new
core is not an extension of the old one: we are not extending the original function by the introduction
of auxiliary variables to reduce clause degrees through new local gadgets; we are changing the function
being represented, and that qualifies the process as a new construction instead of a \emph{cubification}
of the old one. This immediately leads us to our third point: it is not clear how or whether the proofs
we presented in Section~\ref{sec:reduction} would extend to this new environment. Again with some
thought, it is possible to see that the use of the label indexing scheme alone does not guarantee that 
the prime implicates involving variables $v(j)$ can have forms other than $v(j)\lra u(\ell)$, for $j\in[t]$
and $\ell\in L\cup L'$. Moreover, the shape of these prime implicates might depend on the ordering
chosen for the edges of the graph and reintroducing the amplification device based on parameter $d$, 
playing the same role as was before, does not ameliorate the situation (actually, it is completely 
useless). The difficulty in controlling the shape of those prime implicates tarnishes the tie established
between clause minimum prime pure Horn representations and tight, optimal total-covers. It might still 
be possible to replace the clauses of type (d) in a linked-list fashion, but at this point is still not 
clear how to use such approach in a correct and not overly complicated way.

We shall circumvent the above problem through the use of a different structure: we shall replace the
clauses of type (d) by new clauses arranged as \emph{complete binary trees}, i.e., trees in which every
level has all the nodes with the possible exception of the last level, where its nodes are flushed to 
the left; and we shall then link these trees together by their roots. The idea is as follows. For each 
index $i\in[d]$, we will introduce $m-1$ new variables and will arrange them as internal nodes in a complete 
binary tree, where the variables $e(x,y,i)$ will appear as tree leaves. Notice that we will have exactly
$d$ trees. We will then associate each label from $L\cup L'$ with the roots of two of those trees, in an
orderly fashion: the roots (namely, the variables $e(1,i)$) will be seen as nodes and the labels will be 
seen as edges of a path of length $d$. The path will be well defined (i.e., all labels will be reintroduced) 
if all nodes (root variables) are reachable through Forward Chaining from a variable $v(j)$. It is worth
noticing that since $d=1+r\lambda+s\lambda'>|L\cup L'|$, all prime implicates involving the variables $v(j)$
will have the form $v(j)\lra u(\ell)$, for some label $\ell\in L\cup L'$ --- similarly to what happend in
the pure Horn CNF case, we shall show that it is simply not advantageous for these prime implicates to 
have any other form in clause minimum prime pure Horn 3-CNF representations. We now formalize this idea.

Let us rename the labels in $L\cup L'$ as $\ell_\alpha$, $\alpha\in[d-1]$. Let us also index the $m$
edges in $E$ as $e_k$, $k\in[m]$. Let us further introduce new propositional variables $e(k,i)$, 
$k\in[m-1]$, $i\in[d]$, and introduce $e(k,i)$ for $k\in\{m,m+1,\ldots,2m-1\}$ to be an alias to the 
variable $e(x,y,i)$, where $(x,y)=e^{}_{k-m+1}$ is an edge in $E$ according to the indexing above. 
We then create $d$ complete binary trees through the family of clauses:
\[
  \begin{array}{clr}
    \textrm{(d$_1$)} & \displaystyle e(2k,i)\wedge e(2k+1,i)\lra e(k,i) 
                     & \qquad\displaystyle \forall\ k\in[m-1],\ i\in[d]; \\[4mm]

    \textrm{(d$_2$)} & \displaystyle e(1,\alpha)\wedge e(1,\alpha+1)\lra u(\ell_\alpha) 
                     & \qquad\forall\ \alpha\in[d-1]. \\[1mm]
\end{array}
\]

Notice that clauses of type (d$_1$) index the nodes of the tree in a similar way a complete binary
tree is stored inside an \emph{array} (cf. Cormen et al.~\cite{CLRS09}) and that the clauses of type 
(d$_2$) do define the path we mentioned. 

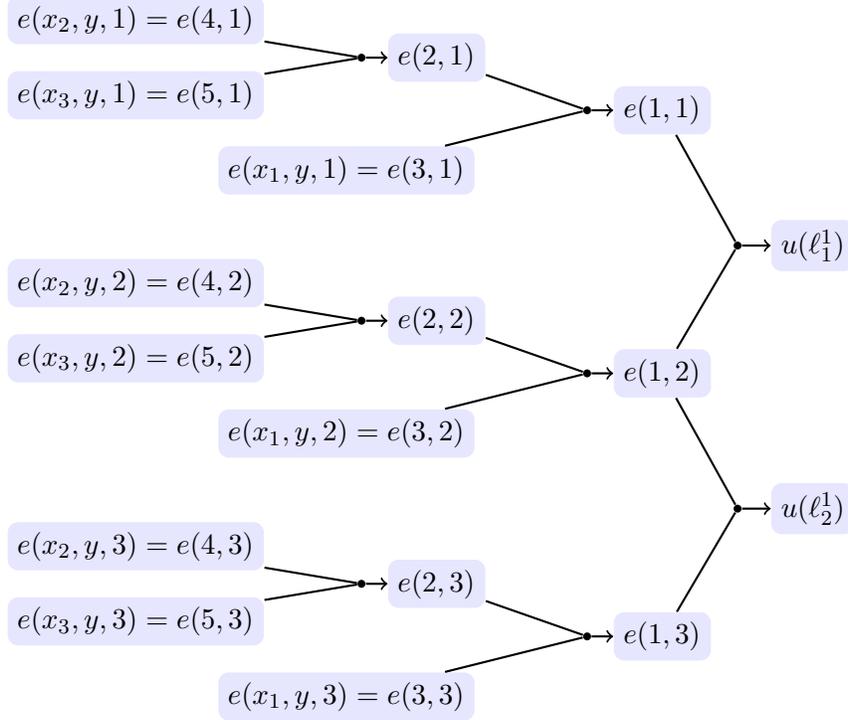
\begin{figure}[bt]
  \centering
  \begin{tikzpicture}[scale=1.0]
    \tikzstyle{dot}    = [fill=black, circle, minimum size=3pt, inner sep=0pt]
    \tikzstyle{vertex} = [fill=blue!10, rectangle, rounded corners]
    \tikzstyle{arc}    = [draw,thick,->,black]
    \tikzstyle{edge}   = [draw,thick,- ,black]

    \node[vertex] (e41) at (0.0,9.0) {$e(x_2,y,1)=e(4,1)$};
    \node[vertex] (e51) at (0.0,8.0) {$e(x_3,y,1)=e(5,1)$};
    \node[vertex] (e31) at (2.8,7.0) {$e(x_1,y,1)=e(3,1)$};

    \node[dot]    (d21) at (3.0,8.5) {}; 
    \node[vertex] (e21) at (4.0,8.5) {$e(2,1)$};
    \node[dot]    (d11) at (6.0,7.8) {}; 
    \node[vertex] (e11) at (7.0,7.8) {$e(1,1)$};
    \draw[edge]   (e41) -- (d21);
    \draw[edge]   (e51) -- (d21);
    \draw[arc]    (d21) -- (e21);
    \draw[edge]   (e21) -- (d11);
    \draw[edge]   (e31) -- (d11);
    \draw[arc]    (d11) -- (e11);

    \node[vertex] (e42) at (0.0,5.5) {$e(x_2,y,2)=e(4,2)$};
    \node[vertex] (e52) at (0.0,4.5) {$e(x_3,y,2)=e(5,2)$};
    \node[vertex] (e32) at (2.8,3.5) {$e(x_1,y,2)=e(3,2)$};

    \node[dot]    (d22) at (3.0,5.0) {}; 
    \node[vertex] (e22) at (4.0,5.0) {$e(2,2)$};
    \node[dot]    (d12) at (6.0,4.3) {}; 
    \node[vertex] (e12) at (7.0,4.3) {$e(1,2)$};
    \draw[edge]   (e42) -- (d22);
    \draw[edge]   (e52) -- (d22);
    \draw[arc]    (d22) -- (e22);
    \draw[edge]   (e22) -- (d12);
    \draw[edge]   (e32) -- (d12);
    \draw[arc]    (d12) -- (e12);

    \node[vertex] (e43) at (0.0,2.0) {$e(x_2,y,3)=e(4,3)$};
    \node[vertex] (e53) at (0.0,1.0) {$e(x_3,y,3)=e(5,3)$};
    \node[vertex] (e33) at (2.8,0.0) {$e(x_1,y,3)=e(3,3)$};

    \node[dot]    (d23) at (3.0,1.5) {}; 
    \node[vertex] (e23) at (4.0,1.5) {$e(2,3)$};
    \node[dot]    (d13) at (6.0,0.8) {}; 
    \node[vertex] (e13) at (7.0,0.8) {$e(1,3)$};
    \draw[edge]   (e43) -- (d23);
    \draw[edge]   (e53) -- (d23);
    \draw[arc]    (d23) -- (e23);
    \draw[edge]   (e23) -- (d13);
    \draw[edge]   (e33) -- (d13);
    \draw[arc]    (d13) -- (e13);

    \node[dot]    (w11) at (8.0,6.0) {};
    \node[vertex] (u11) at (9.0,6.0) {$u(\ell^1_1)$};
    \node[dot]    (w12) at (8.0,2.5) {};
    \node[vertex] (u12) at (9.0,2.5) {$u(\ell^1_2)$};
    \draw[edge]   (e11) -- (w11);
    \draw[edge]   (e12) -- (w11);
    \draw[arc]    (w11) -- (u11);
    \draw[edge]   (e12) -- (w12);
    \draw[edge]   (e13) -- (w12);
    \draw[arc]    (w12) -- (u12);
  \end{tikzpicture}               
  \caption{
    Partial depiction of the complete binary tree chain, which is responsible for reintroducing 
    all the available labels of a \LC instance in our 3-CNF construction, of a toy example
    where the constraint graph is a claw. In such example, the whole chain has nine trees and 
    reintroduces eight different labels. Notice the complete symmetry between the trees.
  }\label{fig:label-tree} 
\end{figure}

We shall illustrate the complete binary tree transformation through the following toy example. 
We start with a \LC instance whose constraint graph is a claw, that is, 
$G=(\{x_1,x_2,x_3\},\{y\},\{(x_1,y),(x_2,y),(x_3,y)\})$, whose label sets are 
$L^{}_0=\{\ell^{}_1,\ell^{}_2\}$ and $ L'_0=\{\ell'_1,\ell'_2\}$, and whose 
constraint set is the union of 
\[
  \Pi^0_{(x_1,y)}=\{(\ell^{}_1,\ell'_1), (\ell^{}_1,\ell'_2)\},\ 
  \Pi^0_{(x_2,y)}=\{(\ell^{}_1,\ell'_2)\},\ \text{and}\  
  \Pi^0_{(x_3,y)}=\{(\ell^{}_2,\ell'_1), (\ell^{}_2,\ell'_2)\}.
\]

The refined \LC instance will then have 8 labels in total 
($L=\{\ell^1_1,\ell^1_2,$ $\ell^2_1,\ell^2_2,\ell^3_1,\ell^3_2\}$ and 
$L'=\{\ell'_1,\ell'_2\}$)\footnote{The refined label $(x_i,\ell_j)\in L$, with $x_i\in X$ and 
$\ell_j\in L^{}_0$, is depicted by $\ell^i_j$ in this example.}
and our CNF construction will introduce $8\times 9=72$ clause of type (d), nine of which have the 
following form
\[
  e(x_1,y,i)\wedge e(x_2,y,i)\wedge e(x_3,y,i)\lra u(\ell^1_1),
\]
as $i\in[d]=[9]$. Clauses of type (d$_1$) replace those subgoals above by
\begin{equation}\label{eq:toy-tree}
  e(4,i)\wedge e(5,i)\lra e(2,i)
  \qquad\text{and}\qquad
  e(2,i)\wedge e(3,i)\lra e(1,i),
\end{equation}
where in this case, $e(x_1,y,i)=e(3,i)$, $e(x_2,y,i)=e(4,i)$, and $e(x_3,y,i)=e(5,i)$.

Finally, clauses of type (d$_2$) link the trees in \eref{eq:toy-tree} together, as e.g.
\[
  e(1,1)\wedge e(1,2)\lra u(\ell^1_1),
  \qquad
  e(1,2)\wedge e(1,3)\lra u(\ell^1_2),
  \qquad
  \textrm{and so on.}
\]
A graphical illustration of part of the above transformation is provided 
in Figure~\ref{fig:label-tree}. 

Now, for clarity purposes, we present this new construction for the pure Horn 3-CNF case in full 
form below:
\[
  \begin{array}{ll}
    \textrm{(a)} & \displaystyle u(\ell)\wedge u(\ell')\lra e(x,y,\ell',i) 
                   \hspace{1.8cm} \forall\ (x,y)\in E,\ (\ell,\ell')\in\Pi_{(x,y)},\ i\in[d]; \\[4mm]

\textrm{(b$_1$)} & \displaystyle \bigwedge_{z\in N(y)}e(z,y,\ell',i)\lra e(x,y,i)
                   \hfill \forall\ (x,y)\in E$, $\ell'\in L'_y$, $i\in[d]$, $\dg(y)\leq 2; \\[6mm]

\textrm{(b$_2$)} & \displaystyle e(z_y^1,y,\ell',i)\wedge e(z_y^2,y,\ell',i)\lra e(1,x,y,\ell',i) \\[.5mm]
                 & \hfill \forall\ (x,y)\in E$, $\ell'\in L'_y$, $i\in[d]$, $\dg(y)\geq 3; \\[2mm]

\textrm{(b$_3$)} & \displaystyle e(z_y^{\beta+2},y,\ell',i)\wedge e(\beta,x,y,\ell',i)\lra e(\beta+1,x,y,\ell',i) \\[.5mm]
                 & \hfill \forall\ (x,y)\in E$, $\ell'\in L'_y$, $i\in[d]$, $\beta\in[\dg(y)-3]; \\[2mm]

\textrm{(b$_4$)} & \displaystyle e(z_y^{\dg(y)},y,\ell',i)\wedge e(\dg(y)-2,x,y,\ell',i)\lra e(x,y,i) \\[.5mm]
                 & \hfill \forall\ (x,y)\in E$, $\ell'\in L'_y$, $i\in[d]$, $\dg(y)\geq 3; \\[4mm]

    \textrm{(c)} & \displaystyle e(x,y,i)\lra e(x,y,\ell',i)
                   \hfill \forall\ (x,y)\in E,\ \ell'\in L'_y,\ i\in[d]; \\[4mm]

\textrm{(d$_1$)} & \displaystyle e(2k,i)\wedge e(2k+1,i)\lra e(k,i)
                   \hfill \forall\ k\in[m-1],\ i\in[d]; \\[4mm]
 
\textrm{(d$_2$)} & \displaystyle e(1,\alpha)\wedge e(1,\alpha+1)\lra u(\ell_\alpha)
                   \hfill \forall\ \alpha\in[d-1]; \\[4mm]

    \textrm{(e)} & \displaystyle v(j)\lra u(\ell)
                   \hfill \forall\ j\in[t],\ \ell\in L\cup L'; \\[2mm]
  \end{array}
\]
where as before, $N(y):=\{x\in X:(x,y)\in E\}$ is the open neighborhood of the vertex $y\in Y$.
Recall that $e(k,i)$ is an alias to a variable $e(x,y,i)$, $(x,y)=e^{}_{k-m+1}$ being an edge
in $E$, if $k\in\{m,m+1,\ldots,2m-1\}$ and that it is a new variable used to build the 
$i$-th tree if $k\in[m-1]$.

\begin{definition}\label{dfi:3fcns}
  Let us call $\Psi$ and $\Phi$ the \emph{canonical} pure Horn 3-CNF formulae defined, respectively, 
  by the families of clauses (a) through (d$_2$) and by all the families of clauses above. Let $g$ 
  and $h$ be, in that order, the pure Horn functions they represent. 
\end{definition}

\subsection{Correctness of the 3-CNF Reduction}
We now proceed to show the correctness of the ideas discussed above. Once again, we will constantly
rely on the canonical representations $\Phi$ and $\Psi$ to make inferences about the functions $h$ 
and $g$ they respectively define. And these inferences will most of the time be made throughout 
Forward Chaining. First, we present the new estimations for the number of clauses and variables
in $\Phi$ and $\Psi$.

\begin{lemma}\label{lem:3cnf-reduction-size}
  Let $d$ and $t$ be positive integers (amplification parameters) and let $r$, $s$, $m$, $\lambda$, 
  $\lambda'$, and $\pi$ be as in Notation~\ref{nta:LC-sizes}. We have the following relations for 
  the number of clauses and variables in $\Phi$, respectively:
  \[
    d(\pi+2m\lambda'+2m-1)-1\leq|\Phi|_c-t(r\lambda+s\lambda')\leq d(\pi+m^2\lambda'+4m)
  \]
  and
  \[
    t\leq|\Phi|_v\leq t+dm(\lambda'+2)+m^2\lambda', 
  \]
  In $\Psi$, those numbers are, respectively, 
  \[
    d(\pi+2m\lambda'+2m-1)-1\leq|\Psi|_c\leq d(\pi+m^2\lambda'+4m)
  \]
  and
  \[
    0 \leq|\Psi|_v\leq dm(\lambda'+2)+m^2\lambda'.
  \]
\end{lemma}
\begin{proof}
  Let $\#(\tilde{\textrm{b}}):=\sum_{i\in[4]}\#(\textrm{b}_i)$ and $\#(\tilde{\textrm{d}}):=\sum_{i\in[2]}\#(\textrm{d}_i)$, 
  where $\#(\alpha)$ denotes the number of clauses of type $(\alpha)$ in $\Phi$.
	
  For each edge $(x,y)\in E$, there was a clause of type (b) whose subgoal had size equal to $\dg(y)$. 
  Each of those clauses was either maintained in case $\dg(y)\leq 2$ (originating the new clauses (b$_1$)) 
  or replaced by the $\dg(y)-1$ new clauses (b$_2$), (b$_3$), and (b$_4$) in a linked-list fashion.
  This procedure results in
  \[
    dm\lambda' \leq \#(\tilde{\textrm{b}}) = d\lambda'\sum_{(x,y)\in E}(\dg(y)-1) \leq dm(m-1)\lambda'.
  \]

  The clauses of type (d) were replaced by clauses of type (d$_1$) and (d$_2$). The new clauses of type 
  (d$_1$) describe $d$ complete binary trees, each of which having $2m-1$ nodes and hence, height 
  $\eta=1+\floor{\log (2m-1)}\leq\floor{\log 4m}$. Considering also the $d-1$ new clauses of type (d$_2$),
  we then obtain that
  \[
    d(2m-1)-1\leq \#(\tilde{\textrm{d}}) = d-1 + d\sum_{l=1}^{\eta-1} 2^l = d-1 + d(2^\eta-1) \leq 4dm.
  \]

  The new gadgets introduces $d(m-1)$ new variables $e(k,i)$ (notice the aliases are not new variables) 
  and at most $dm(m-1)\lambda'$ new variables $e(\beta,x,y,\ell',i)$. Now, the results follow by using 
  the remaining estimates of Lemma~\ref{lem:cnf-reduction-size} and by recalling that $\Phi$ and $\Psi$ 
  differ only on the clauses of type (e).
\end{proof}

Similarly to the pure Horn CNF case, as long as the quantities $d$ and $t$ are polynomial in $s$, namely, 
the number of vertices in $Y$, the new construction of $\Phi$ from $\cL$ can also be carried out in polynomial 
time in $s$ or, in another way, in polynomial time in the number of variables of $h$ (cf. 
Remark~\ref{rem:DS-inst-size}).

The same arguments used to prove Lemma~\ref{lem:g-excl-comp} apply in this new setting as the differences 
introduced by the new clauses are of a local nature. Specifically, it is still immediate that no clause in 
$\Psi$ has head outside $V_{g}$, the set of variables occurring in $\Psi$, and hence, that $V_g$ is closed 
under Forward Chaining in $\Phi$. Thus, the set 
$
  \cX(V_g) := \{C\in\cI(h): \mathsf{Vars}(C)\subseteq V_g\}
$ 
is still an exclusive family for $h$ (cf. Lemma~\ref{lem:FC-exclusive-set}) and $g\equiv\Psi=\Phi\cap\cX(V_g)$ 
is an $\cX(V_g)$-exclusive component of $h$. We have just proved the following.

\begin{lemma}\label{lem:3cnf-g-excl-comp}
  The new function $g$ is an exclusive component of the new function $h$. Consequently, $g$ can be analysed 
  and minimized separately. \qed
\end{lemma}

It is not hard to see that the bounds provided by Lemma~\ref{lem:cnf-psi-lb} are still valid in this new 
setting. Furthermore, Lemmas~\ref{lem:FC(v(j))} and~\ref{lem:v(j)-quadratic} transfer in an almost verbatim 
fashion: their proof can be quickly adapted as the new clauses and variables do not disrupt any of the 
conclusions obtained.

In the CNF construction presented in the previous section, all clauses of type (d) had the same subgoal.
As explained in the beginning of this section, those clauses' subgoals were potentially long and we had
to replace them by $d$ gadgets whose structure mimics those of complete binary trees. It is still true
that if one label in $L\cup L'$ is reintroduced by a clause of type (d$_2$), then so are all the remaining
others. Like before, the reason is the high degree of symmetry occurring inside the exclusive component of 
$g$. Nevertheless, as the 3-CNF construction is more involved, we shall still present below complete proofs 
for the analogues of Lemmas~\ref{lem:cnf-no-shortcuts}, \ref{lem:cnf-tightness-up}, and 
\ref{lem:cnf-tightness-down}.

\begin{lemma}[Analogue of Lemma~\ref{lem:cnf-no-shortcuts}]\label{lem:3cnf-no-shortcuts}
  Let $d=1+r\lambda+s\lambda'$. In any clause minimum prime pure Horn 3-CNF representation of $h$, the prime implicates 
  involving the variables $v(j)$ have the form $v(j)\lra u(\ell)$, for all indices $j\in[t]$, and for some labels 
  $\ell\in L\cup L'$.
\end{lemma}
\begin{proof}
  Let $\Upsilon=\Theta\wedge\Gamma$ be a clause minimum prime pure Horn 3-CNF representation of $h$, with $\Theta$ being 
  a clause minimum pure Horn 3-CNF representation of $g$. According to an analogue of Lemma~\ref{lem:v(j)-quadratic}, 
  all prime implicates of $h$ involving the variables $v(j)$ are quadratic. So, for all indices $j\in[t]$ and all indices 
  $i\in[d]$ define the sets
  \begin{align*}
    \Gamma^j_0 & := \Gamma\cap\{v(j)\lra u(\ell): \ell\in L\cup L'\},\\
    \Gamma^j_i & := \Gamma\cap\{v(j)\lra e(k,i),\\
               & \phantom{:= \Gamma\cap\{}\; v(j)\lra e(x,y,\ell',i),\\
               & \phantom{:= \Gamma\cap\{}\; v(j)\lra e(\beta,x,y,\ell',i) :
                 k\in[2m-1], (x,y)\in E, \ell'\in L'_y, \beta\in[\dg(y)-2]\}.
  \end{align*}
  
  Recall that $e(k,i)$ is an alias to $e(x,y,i)$ if $(x,y)=e^{}_{k-m+1}$ and $k\in\{m,m+1,\ldots,2m-1\}$.
  Our goal is to show that the chosen value for the parameter $d$ forces all the sets $\Gamma^j_i$ to be
  simultaneously empty and consequently, that all the prime implicates involving the variables $v(j)$
  in clause minimum pure Horn CNF representations of $h$ have the claimed form. We shall accomplish this 
  in two steps. 

  Let $j\in[t]$. We first show that if a set $\Gamma^j_i\neq\emptyset$ for some index $i\in[d]$, then 
  $\Gamma^j_i\neq\emptyset$ for all indices $i\in[d]$, simultaneously. 

  \setcounter{claim}{-1}
  \begin{claim}
    It holds that
    \[
      e(k,i) \in F_{\Upsilon} (\{v(j)\})
      \quad\text{if and only if}\quad
      e(k,i') \in F_{\Upsilon} (\{v(j)\}),
    \]
    for all indices $i,i'\in[d]$, with $i\neq i'$.
  \end{claim}
  \begin{proof}
    The families of clauses (a), (b$_1$), (b$_2$), (b$_3$), (b$_4$), and (c) are completely symmetric with respect 
    to the indexing variable $i$ and do not interfere with each other. That is, for $i_1\neq i_2$, variables depending 
    upon $i_1$ do not trigger clauses indexed by $i_2$ during Forward Chaining, and vice-versa.
  \end{proof}

  \begin{claim}
    All clauses of type (d$_2$) have two variables in their subgoals which are roots of different, but 
    completely symmetric trees (the trees specified by the clauses of type (d$_1$)). Therefore, because
    of this symmetry and Claim (0), during the execution of the Forward Chaining procedure from $\{v(j)\}$, 
    either they all clauses of type (d$_2$) trigger simultaneously or none of them do. The reason for them 
    not to trigger is the absence of variables $e(k,i)$, for some $k\in[2m-1]$ and all $i\in[d]$, in the 
    Forward Chaining closure from $\{v(j)\}$, i.e, $e(k,i)\not\in F_{\Upsilon}(\{v(j)\})$.
  \end{claim}
  \begin{proof}
    A simple inspection of the families of clauses shows that Claim (1) holds. 
  \end{proof}

  Now, for each index $i\in[d]$, let $\Xi^{}_i$ be the collection of clauses of types (a), (b$_1$), (b$_2$),
  (b$_3$), (b$_4$), and (c) that depend on $i$.

  \begin{claim}
    It holds that
    \[
      e(k,i) \in F_{\Gamma^j_0\cup\Xi^{}_i} (\{v(j)\})
      \quad\text{if and only if}\quad
      e(k,i') \in F_{\Gamma^j_0\cup\Xi^{}_{i'}} (\{v(j)\}),
    \]
    for all indices $i,i'\in[d]$, with $i\neq i'$.
  \end{claim}
  \begin{proof}
    Notice that the families of clauses (a), (b$_1$), (b$_2$), (b$_3$), (b$_4$), and (c) are completely symmetric 
    with respect to the indexing variable $i$. Moreover, for $i_1\neq i_2$, the clauses indexed by $i_1$ do 
    not interfere with the clauses indexed by $i_2$ during an execution of the Forward Chaining procedure. 
    In other words, variables depending upon $i_1$ do not trigger clauses indexed by $i_2$, and vice-versa. 
    These two properties, symmetry and non interference, proves Claim (2). 
  \end{proof}

  \begin{claim}
    If there is a variable $e(k,i)$, with $k\in[2m-1]$ and $i\in[d]$, such that
    \[
      e(k,i)\not\in F_{\Gamma^j_0\cup\left(\bigcup_{i\in[d]}\Xi_i^{}\right)}(\{v(j)\})
    \]
    then
    \[
      e(k,i)\not\in F_{\Gamma^j_0\cup\left(\bigcup_{i\in[d]}\Xi_i^{}\right)\cup\left(\bigcup_{i'\neq i}\Gamma^j_{i'}\right)}
      (\{v(j)\}).
    \]
    Moreover, this implies that $\Gamma^j_i\neq\emptyset$.
  \end{claim}
  \begin{proof}
    The symmetry and non interference properties of families of clauses (a), (b), and (c) also justifies the first
    part of Claim (3). To see it, just notice that were the claim to be false, the prime implicates in $\Gamma^j_{i'}$ 
    would be trigging clauses involving the variable $e(k,i)$ in an execution of the Forward Chaining procedure. 
    Since $i'\neq i$, this cannot happen. The second part follows immediately from the validity of the first part
    together with the fact that $\Upsilon$ represents $h$. 
  \end{proof}

  To finish the first step, notice that since Claim (3) is valid for any $i\in[d]$, Claim (2) implies that if 
  $\Gamma^j_i\neq\emptyset$ for some index $i\in[d]$, then $\Gamma^j_i\neq\emptyset$ for all indices $i\in[d]$, 
  simultaneously.

  For the second step, suppose that $\Gamma^j_i\neq\emptyset$ for all indices $i\in[d]$. We then have that
  \[
    \gamma:=\sum_{i\in[d]}|\Gamma^j_i|\geq d=1+r\lambda+s\lambda'=1+|L\cup L'|, 
  \]
  that is, $\gamma$ is strictly larger than the number of all available labels in $\cL$. This implies that 
  the following pure Horn CNF
  \begin{align*}
    \Delta_j & := \left(\Upsilon\setminus\bigcup_{i\in d}\Gamma^j_i\right)\cup\Big\{v(j)\lra u(\ell):\ell\in L\cup L'\Big\}\\
             &\ = \Theta\cup\left(\left(\Gamma\setminus\bigcup_{i\in d}\Gamma^j_i\right)\cup
                  \Big\{v(j)\lra u(\ell):\ell\in L\cup L'\Big\}\right),
  \end{align*}
  has fewer clauses than $\Upsilon$ (or more precisely, it implies that $|\Delta_j|_c <= |\Upsilon|_c - 1$). 

  Now, since that $\Theta$ is a (clause minimum) representation of the exclusive component $g$, and that the
  set of clauses $\{v(j)\lra u(\ell):\ell\in L\cup L'\}$ makes all available labels reachable by Forward
  Chaining from $\{v(j)\}$, it follows that $F_{\Delta_j}(\{v(j)\})=V_g\cup\{v(j)\}$. Furthermore, the change
  in clauses did not influence the Forward Chaining procedure from any other variable (other than $v(j)$), 
  and thus $F_{\Delta_j}(\{w\})=F_{\Upsilon}(\{w\})$ for all variables $w\neq v(j)$. Thus, an analogue of
  Lemma~\ref{lem:FC(v(j))} implies that $\Delta_j$ is a representation of $h$.

  We then have that $\Delta_j$ is a shorter representation for $h$, contradicting the optimality of $\Upsilon$.
  Therefore, the sets $\Gamma^j_i=\emptyset$ for all indices $i\in[d]$. As the above arguments do not depend on 
  any particular value of $j$, they can be repeated for all of them.
\end{proof}

\begin{definition}\label{dfi:3cnf-tcover}
  Let $d=1+r\lambda+s\lambda'$ and let $\Upsilon$ be a prime and 
  irredundant pure Horn 3-CNF representation of $h$. For each $j\in[t]$, consider the set 
  \[
    S_j=\left\{\ell\in L\cup L':v(j)\lra u(\ell) \in \Upsilon\right\}
  \] 
  and define the function  $f_j:X\ra L, Y\ra L'$ given by $f_j(x)=S_j\cap L_x$ for vertices $x\in X$ and 
  $f_j(y)=S_j\cap L'_y$ for vertices $y\in Y$.
\end{definition}

\begin{lemma}[Analogue of Lemma~\ref{lem:cnf-tightness-up}]\label{lem:3cnf-tightness-up}
  Let $\Upsilon$ be as in the above Definition. For all indices $j\in[t]$ and vertices $y\in Y$, 
  it holds that $|f_j(y)|\leq 1$.
\end{lemma}
\begin{proof}
  Let $\Upsilon$ be as in Definition~\ref{dfi:3cnf-tcover} and suppose indirectly that the claim is false, 
  that is, there is an index $j\in[t]$ and a vertex $y\in Y$ such that $|f_j(y)| > 1$.

  During the proof, recall that the chosen value for the parameter $d$ implies, according to 
  Lemma~\ref{lem:3cnf-no-shortcuts}, that all prime implicates of $\Upsilon$ involving the variable 
  $v(j)$ must have the form $v(j)\lra u(\ell)$, with $\ell\in L\cup L'$.

  Let $\ell'\in f_j(y)$ and define the expression 
  \[
    \Upsilon':=\Upsilon\setminus\{v(j)\lra u(\ell')\}.
  \]

  It is enough to show that $F_{\Upsilon'}(\{v(j)\})=V_g\cup\{v(j)\}$, that is, that $\Upsilon'$ is also a 
  representation of $h$ (by an analogue of Lemma~\ref{lem:FC(v(j))}). Suppose that is not the case. Since 
  $\Upsilon$ and $\Upsilon'$ differ only in the clause $v(j)\lra u(\ell')$, it must be the case that 
  $u(\ell')\not\in F_{\Upsilon'}(\{v(j)\})$. This happens if the clause of type (d$_2$) associated to $u(\ell')$ 
  is not trigged. For this to occur, there must be an index $k\in[2m-1]$ and an index $i\in[d]$ such 
  that $e(k,i)\not\in F_{\Upsilon'}(\{v(j)\})$. 

  Now, for $y$ and $i$ as above, notice that: 
  (i) the variable $e(k,i)$ would be included in $F_{\Upsilon'}(\{v(j)\})$ as long as there were a label in 
  $L'_y$ such that the corresponding clause of type (b$_1$) or clauses of types (b$_2$), (b$_3$), and (b$_4$) 
  were trigged; and
  (ii) once such clauses of type (b$_1$)--(b$_4$) were trigged, the appropriated clauses of type (c) would 
  trigger, thus making the other clauses of type (b$_1$)--(b$_4$) associated to $y$ and $i$ to also trigger. 

  Therefore, for $e(k,i)=e(x,y,i)$ to not belong to $F_{\Upsilon'}(\{v(j)\})$, it must be the case that for every 
  label $\ell''\in f_j(y)\setminus\{\ell'\}$ there exists a vertex $z(\ell'')\in N(y)$ for which
  \[
    e(z(\ell''),y,\ell'',i)\not\in F_{\Upsilon'}(\{v(j)\}).
  \]
  
  For this latter relation to be true, we must have that the clauses 
  \begin{equation}\label{eq:3cnf-trigged}
    u(\ell)\wedge u(\ell'')\lra e(z(\ell''),y,\ell'',i)
  \end{equation}
  are not trigged in the Forward Chaining procedure on $\Upsilon'$ starting with $\{v(j)\}$, for every 
  label $\ell\in f_j(z(\ell''))$ with $(\ell,\ell'')\in\Pi_{(z(\ell''),y)}$. 

  However, according to Definition~\ref{dfi:3cnf-tcover}, for each label $\ell''\in f_j(y)\setminus\{\ell'\}$
  and each label $\ell\in f_j(z(\ell''))$, there are clauses $v(j)\lra u(\ell'')$ and $v(j)\lra u(\ell)$,
  respectively, in $\Upsilon$ and, consequently, in $\Upsilon'$. This implies that the clauses \eref{eq:3cnf-trigged} 
  are trigged, which implies that $u(\ell')\in F_{\Upsilon'}(\{v(j)\})$, which then implies that $\Upsilon'$ is 
  also a representation of $h$. Since this contradicts the irredundancy of $\Upsilon$, it follows that 
  $|f_j(y)|\leq 1$, thus concluding the proof.
\end{proof}

\begin{lemma}[Analogue of Lemma~\ref{lem:cnf-tightness-down}]\label{lem:3cnf-tightness-down}
  Let $\Upsilon$ be a clause minimum prime pure Horn 3-CNF of $h$. Then it is prime and irredundant, 
  so Definition~\ref{dfi:3cnf-tcover} applies. We claim that for all indices $j\in[t]$ and vertices $y\in Y$, 
  it holds that $|f_j(y)| \geq 1$.
\end{lemma}
\begin{proof}
  Suppose that the claim is false, that is, 
  there is an index $j\in[t]$ and a vertex $y\in Y$ such that $|f_j(y)| = 0$.

  Then clauses $v(j)\lra u(\ell')$, for all labels $\ell'\in L'_y$, are absent from 
  $\Upsilon$. Recall that the chosen value for the parameter $d$ implies that all prime implicates 
  of $\Upsilon$ involving $v(j)$ are quadratic (Lemma~\ref{lem:3cnf-no-shortcuts}).

  Thus, no clause of type (a) dependent on the vertex $y$ is trigged during a Forward 
  Chaining from $\{v(j)\}$ and hence, no clauses of type (b$_1$), (b$_2$), (b$_3$), (b$_4$), and 
  (c) dependent on $y$ are trigged either. This gives that the variables $e(k,i)$, for all 
  indices $k\in\{m,m+1,\ldots,2m-1\}$ such that $e_{k-m+1}=(x,y)$ (recall the indexing of the
  edges in $E$) and $x\in N(y)$, and all indices $i\in[d]$, do not belong to the Forward Chaining 
  closure (from $\{v(j)\}$). This implies further that no variables $e(1,i)$ belong to the Forward 
  Chaining closure. Therefore, no clause of type (d$_2$) is trigged and no label $\ell'\in L'_y$ is 
  reintroduced. In other words, it is the case that $u(\ell')\not\in F_{\Upsilon}(\{v(j)\})$ and hence, that 
  $F_{\Upsilon}(\{v(j)\}) \neq F_h(\{v(j)\})$. By an analogue of Lemma~\ref{lem:FC(v(j))}, $\Upsilon$ 
  does not represent $h$, a contradiction. So, it must be the case that $|f_j(y)|\geq 1$.
\end{proof}

Once again, combining the two lemmas above, gives the following tight result.
\begin{corollary}\label{cor:3cnf-tightness}
  Let $\Upsilon$ be a clause minimum prime pure Horn 3-CNF of $h$. Then, for all indices $j\in[t]$ 
  and vertices $y\in Y$, it holds that $|f_j(y)| = 1$.
  \qed
\end{corollary}

Analogues of Lemmas~\ref{lem:cnf-tcover}, \ref{lem:cnf-represents}, and \ref{lem:cnf-optimality}
can also be obtained in the same semi-verbatim fashion, and Remark~\ref{rem:diff-tcovers} remains 
valid in this context. Hence, we still obtain the following.

\begin{corollary}\label{cor:3cnf-upsilon-ub}
  Let $\Psi$ be as in Definitions~\ref{dfi:3fcns} and $\Upsilon$ be a clause minimum prime pure Horn 
  3-CNF of $h$. It holds that 
  \[
    |\Psi|_c/(\lambda+\lambda')\leq |\Upsilon|_c - t(\kappa(f)r + s)\leq~|\Psi|_c,
  \] 
  where $\kappa(f)r + s$ is the total number of labels in a tight optimal total-cover $f$ for $\cL$.
  \qed
\end{corollary}

We are also able to claim the following result.

\begin{theorem}\label{thm:3cnf-hardness}
  Let $c$ be a fixed constant close to $1/2$. Unless $\PT=\NP$, the minimum number of clauses of a pure Horn 
  function on $n$ variables cannot be approximated in polynomial time (dependent on $n$) to within a factor of 
  \[ 
    \rho_c(n^{\varepsilon})\geq 2^{\varepsilon(\log n)^{1-1/\delta_c(n)}}=2^{\log^{1-o(1)} n},
  \]
  where $\delta_c(n)=(\log\log n)^{c}$, even when the input is restricted to 3-CNFs with $O(n^{1+2\varepsilon})$ 
  clauses, for some $\varepsilon\in(0,1/6]$.
\end{theorem}
\begin{proof}
  The proof follows closely the one given for Theorem~\ref{thm:cnf-hardness}, just using the estimates provided by
  Lemma~\ref{lem:3cnf-reduction-size} instead of the ones in Lemma~\ref{lem:cnf-reduction-size}.
\end{proof}

\subsection{Number of Literals}
With the exception of the variables $v(j)$, with $j\in[t]$, that only appear as subgoals in quadratic prime 
implicates, every other variable appears in subgoals and heads of mostly cubic pure Horn clauses. The functions 
we are dealing with are pure Horn and therefore, have no unit clauses. So, it is the case that
$2|\Phi|_c\leq |\Phi|_l\leq 3|\Phi|_c$ or in other words, that $|\Phi|_l=\Theta(|\Phi|_c)$, where $\Phi$ is a 
pure Horn 3-CNF formula obtained from our 3-CNF construction above. We then have the following result.

\begin{corollary}\label{cor:lit-hardness}
  Unless $\PT=\NP$, the minimum number of literals of a pure Horn function on $n$ variables cannot be
  approximated in polynomial time (in $n$) to within a factor of $2^{\log^{1-o(1)} n}$, even
  when the input is restricted to 3-CNFs with $O(n^{1+2\varepsilon})$ clauses, for some $\varepsilon\in(0,1/6]$.
  \qed
\end{corollary}

\section{Sub-exponential Time Hardness Results}\label{sec:subexp}
The hardness of approximation results of Sections~\ref{sec:reduction} and \ref{sec:3-cnf} apply to the
scenario where the amount of available computational power is polynomial in $n$, the number of variables 
of a pure Horn function. In this section, we extend those results by showing that even when the computational 
power available is sub-exponential in $n$, it still not likely to be able to obtain a constant factor approximation
for such problems. The main ingredients of this section are: a stronger complexity theoretic hypothesis,
a new \LC result, and our pure Horn 3-CNF construction.

Recall that $k$-SAT is the problem of determining if a $k$-CNF formula, that is, one in which each and every
clause has at most $k$ literals, has a satisfying assignment of Boolean values to it variables. The following 
conjecture, called \emph{Exponential Time Hypothesis (ETH)}, concerns the time solvability of the $k$-SAT 
problem and was introduced by Impagliazzo and Paturi~\cite{IP01}.
\begin{conjecture}[Impagliazzo and Paturi~\cite{IP01}]\label{con:ETH}
  For $k\geq 3$, define $s_k$ to be the infimum of the set
  \begin{equation*}
    \Big\{\delta : \text{there exists an $O\Big(2^{\delta n}\Big)$ time algorithm
      for solving the $k$-SAT problem}\Big\},
  \end{equation*}
  with $n$ being the number of variables of the $k$-SAT instance.
  The \emph{Exponential Time Hypothesis (ETH)} states that $s_k>0$ for $k\geq 3$.
\end{conjecture}

In other words, if true, the ETH implies that there is no sub-exponential time algorithm for $k$-SAT with 
$k\geq 3$, what in turn implies that $\PT\neq\NP$. The converse of this last implication however, does not
hold and this establish ETH as a stronger hypothesis. It has many implications beyond search and optimization
problems, e.g. in communication, proof, and structural complexity, and it is widely believed to be true.

In Section~\ref{sec:LC}, we adressed the minimization flavor of the \LC problem and briefly mentioned the 
existence of a maximization counterpart. Since the \LC results we shall use in this section were originally
obtained in the maximization setting, we introduce it below together with a ``weak duality'' type of result 
that binds both flavors.

\begin{definition}\label{dfi:packing-labeling}
  Let $\cL^{}_0=(G,L_0^{},L_0',\Pi^{}_0)$ be a \LC instance and $f^{}_0$ be a labeling for it, as in 
  Definitions~\ref{dfi:LC-inst} and \ref{dfi:LC-total-cover}, respectively. Let us call $f^{}_0$ \emph{packing} 
  if it assigns exactly one label per vertex of $G$. A packing labeling is \emph{optimal} if it maximizes 
  the fraction of covered edges of $G$. We denote this maximum fraction by $\mu(\cL^{}_0)$.
\end{definition}

Notice that differently from a total-cover, a packing labeling does not necessarily covers all the 
edges of the graph $G$ and thus, $0< \mu(\cL^{}_0)\leq 1$. The strict lower bound is due to non-empty
relations of admissible pair of labels in $\Pi^{}_0$ --- in the maximization setting, these relations
are also called \emph{projections} as they can be interpreted (or redefined) as mappings
$L^{}_0\longmapsto L'_0$.

In this section, we shall assume that the connected bipartite graph $G=(X,Y,E)$ is regular, that is, 
each vertex $x\in X$ has degree $d_X\geq 1$ and each vertex $y\in Y$ has degree $d_Y\geq 1$. This 
assumption is without loss of generality since \LC instances can be regularized without significantly
altering their sizes and promises (cf. Dinur and Harsha~\cite{DH13}). Moreover, the instances
occurring in the hardness theorem we shall use in this section are regular.

\begin{problem}\label{pro:label-cover}
  Let $0<\xi\leq 1$ be any fixed constant. A \LC instance $\cL^{}_0$ has \emph{packing-promise} $\xi$ if either
  $\mu(\cL^{}_0)=1$ or $\mu(\cL^{}_0)\leq\xi$. That is, either all edges of $\cL^{}_0$ can be covered or at
  most a $\xi$ fraction of them can. The $\LCM_\xi$ problem is a promise problem which receives a \LC 
  instance with packing-promise $\xi$ as input and correctly classify it in one of those two cases.
\end{problem}

The behavior of $\LCM_\xi$ is left unspecified for non-promise instances and any answer 
is acceptable in that case. As before, the definitions above can be easily extended to refined \LC instances.

The following result appears in Arora and Lund~\cite{AL96} and provides a link between the two flavors of
the \LC problem. In a nutshell, it implies that gap producing reductions from \NP-complete problems like
3-SAT to the maximization flavor can be viewed as reductions to the minimization flavor as well. Hence,
hardness of approximation results can be transfered from one flavor to the other. We decided to include its 
proof below for clarity reasons. 

\begin{lemma}[Arora and Lund~\cite{AL96}]\label{lem:LC-weakduality}
  ``Weak Duality:'' for any (refined) \LC instance $\cL$, we have that
  \[
    \mu(\cL) \geq \frac{1}{\kappa(\cL)}.
  \]
  That is, the reciprocal of the value of an optimal packing labeling lower bounds the cost of an optimal 
  (tight) total-cover.
  Furthermore, if $\cL$ has packing-promise $\xi$, then it has covering-promise $\rho\geq 1/\xi$.
\end{lemma}
\begin{proof}
  Let $f$ be an optimal total-cover for $\cL$ of cost $\kappa(\cL)$. That is, $f$ covers all the edges of 
  the graph $G=(X,Y,E)$ assigning $|f(z)|$ labels to each vertex of $z\in X\cup Y$. For simplicity reasons
  and without loss of generality, suppose that $f$ is tight. Recall that by definition, $\kappa(\cL)$ is 
  the average number of labels assigned by $f$ to the vertices in $X$, namely,
  \begin{equation}\label{eq:numlabels}
    \sum_{x\in X} |f(x)| = \kappa(\cL)\cdot|X|.
  \end{equation}

  Consider the following randomized procedure: for each vertex $x\in X$, pick a label at random in 
  $f(x)$ and delete the remaining ones. Let $f'$ be the resulting labeling. As $|f'(z)|=1$ for all 
  vertices $z\in X\cup Y$, $f'$ is a packing labeling for $\cL$ and the expected fraction of edges 
  covered in $f'$ is a lower bound for $\mu(\cL)$.

  Let $\ell\in f(x)$ be a label used in $f$ to cover an edge $(x,y)$. The probability that $\ell\in f'(x)$,
  namely, that it survived the deletion process and ended up in $f'$ is $1/|f(x)|$. The expected number 
  of edges of $G$ still covered in $f'$ is then at least
  \begin{equation}\label{eq:numedgescovered}
    \sum_{(x,y)\in E} \frac{1}{|f(x)|} 
      = \sum_{x\in X} \frac{d_X}{|f(x)|} 
      \geq d_X \frac{|X|^2}{\displaystyle \sum_{x\in X} |f(x)|}
      = d_X\frac{|X|^2}{\kappa(\cL)\cdot|X|}
      = \frac{|E|}{\kappa(\cL)},
  \end{equation}
  where in the first and last equalities, we used the assumption that $G$ is regular and thus, has 
  $|E|=d_X\cdot|X|$ edges; in the inequality, we used the fact that $\sum_x 1/|f(x)|$ is minimized
  when the values $|f(x)|$ are all equal; and in the second to last equality we used 
  Equation~\eref{eq:numlabels}. 

  The above randomized procedure thus gives a packing labeling whose expected fraction of edges covered 
  is at least $1/\kappa(\cL)$. So, there must exist a packing labeling attaining at least this fraction 
  of edges covered and therefore, $\mu(\cL) \geq 1/\kappa(\cL)$ as claimed. The relation on the promise
  bounds follows in a similar way from Equation~\eref{eq:numedgescovered}.
\end{proof}

Moshkovitz and Raz~\cite{MR-10-JACM}, in a celebrated breakthrough, introduced a new two-query projection
test Probabilistically Checkable Proof system with sub-constant error and quasi-linear size. Essentially,
they started with a 3-SAT CNF instance $\phi$ of size (number of clauses) equal to $\sigma$ and showed that
it is \NP-hard to solve $\LCM_{1/\rho}$, for some $\rho=\rho(\sigma)$. Moreover, their reduction produces a 
graph of size at most $\sigma^{1+o(1)}$ and uses fixed label sets $L$ and $L'$ whose 
sizes depend on the value of the promise $\rho$. A simpler proof was later found by Dinur and Harsha~\cite{DH13} 
who started their reduction from a related but slightly different problem (satisfiability of circuits 
instead of formulas), brought different techniques to the mix, and formally stated the result as follows.

\begin{theorem}[Moshkovitz and Raz~\cite{MR-10-JACM}, Dinur and Harsha~\cite{DH13}]\label{thm:LC-Dinur-Harsha}
  There exist constants $c>0$ and $0<\beta<1$ such that for every function 
  $1<\rho(\sigma)\leq 2^{O(\log^\beta\sigma)}$, the following statement holds:

  There exists label sets $L$ and $L'$ of sizes $\exp({\rho(\sigma)^c})$ and $O(\rho(\sigma)^c)$, 
  respectively, such that it is \NP-hard to solve $\LCM_{1/\rho(\sigma)}$ over these label sets. 
  Furthermore, the size of the constraint graph of the \LC instance produced by this reduction 
  is at most $\sigma\cdot 2^{O(\log^\beta \sigma)}\cdot (\rho(\sigma))^c =\sigma^{1+o(1)}$.
\end{theorem}

Notice that differently to what happens in the reduction of Theorem~\ref{thm:LC-inapprox} where the
label sets are instance dependent, in the above theorem the label sets are fixed and their sizes depend
on the promise value (also called soundness value in this context). Also, the size of label set $L$ is
polynomial if $1< \rho(\sigma) \leq \textrm{polylog}(\sigma)$ and super-polynomial if
$\textrm{polylog}(\sigma) < \rho(\sigma)\leq 2^{O(\log^\beta\sigma)}$. It is worth mentioning that these
differences do not affect our constructions (as they scale up appropriately), but do require changes
in the calculations when showing our hardness results. 

Using Lemma~\ref{lem:LC-weakduality}, we can apply the above hardness result to the minimization flavor 
of the \LC problem with promise $\rho(\sigma)$. While the resulting hardness factor is smaller than the
one given by Dinur and Safra~\cite{DS04} in Theorem~\ref{thm:LC-inapprox}, the quasi-linear size of
the constraint graph (and hence of the instance, as the label sets are fixed) allows the following: 

\begin{corollary}[Moshkovitz and Raz~\cite{MR-10-JACM}]\label{cor:LC-subexp-time}
  Assuming the Exponential Time Hypothesis (cf. Conjecture~\ref{con:ETH}, i.e., 3-SAT requires 
  $\exp(\Omega(\sigma))$ time to be solved), quasi-polynomially sized instances of $\LCM_{1/\rho(\sigma)}$ 
  and of $\LC_{\rho(\sigma)}$ cannot be solved in less than $\exp\bigl(\sigma^{1-o(1)}\bigr)$ 
  time.
  \qed
\end{corollary}

This can be used to rule out better approximations in sub-exponential time for other problems by 
further reductions from the \LC problem.

\begin{remark}\label{rem:new-LC-param}
  The new hardness result above brings along a new parametrization, which is slightly 
  different from the one given in Remark~\ref{rem:DS-inst-size}. We now have that:
  $s=\sigma^{1+o(1)}$ with $o(1)\approx(\log\log\sigma)^{-\Omega(1)}$, 
  $r=\sigma 2^{O(\log^\beta\sigma)}=\sigma^{1+o(1)}$, 
  $m=r\rho(\sigma)^c=\sigma^{1+o(1)}$, 
  $\lambda=O\bigl(2^{\rho(\sigma)^c}\bigr)$, and 
  $\lambda'=O(\rho(\sigma)^c)$.
\end{remark}

Let $\vartheta>0$ be such that $c\vartheta\leq 1$ and take $\rho(\sigma)=\log^{\vartheta}\sigma$.
We then have that 
\[
  \lambda =O\bigl(2^{(\log^{\vartheta}\sigma)^c}\bigr)=O(\sigma)\qquad \text{and}\qquad 
  \lambda'=O\bigl((\log^{\vartheta}\sigma)^c\bigr)  =O(\log\sigma),
\]
and since $m\leq\pi\leq m\lambda\lambda'$, we also have that 
\[
  2^{O(\log^\beta\sigma)}\sigma\log^{c\vartheta}\sigma\leq\pi = O\left(2^{O(\log^\beta\sigma)}\sigma^2\log^2\sigma\right).
\]
Furthermore, it follows that 
\begin{equation}\label{eq:d-subexp}
  d = 1+r\lambda+s\lambda' = \Theta\left(2^{O(\log^\beta\sigma)}\sigma^2+\sigma^{1+o(1)}\log^{c\vartheta}\sigma\right).
\end{equation}
Combining all the above, we obtain that as long as the (gap) amplification device $t$ is polynomial in 
the number of variables of the canonical pure Horn 3-CNF $\Phi$, the construction of $\Phi$ can still be
carried out in polynomial time in this setting.

We are now ready to show a hardness of approximation result for the case when sub-exponential time is allowed. 
The proof closely follows the one presented for the polynomial time setting. 

\begin{theorem}\label{thm:cnf-hardness-subexp}
  Assuming the ETH, the minimum number of clauses and literals of pure Horn functions in $n$ variables cannot be
  approximated in $\exp(n^{\delta})$ time, for some $\delta\in(0,1)$, to within factors of $O(\log^{\vartheta} n)$
  for some $\vartheta>0$, even when the input is restricted to 3-CNFs with $O(n^{1+\varepsilon})$ clauses
  and $\varepsilon>0$ some small constant.
\end{theorem}
\begin{proof}
  Let $\cL_0$ be a \LC promise instances in compliance to Theorem~\ref{thm:LC-Dinur-Harsha} and let $\cL$ be its 
  refinement. Let $d$ be as in Equation~\eref{eq:d-subexp}, $\Phi$ be the canonical pure Horn 3-CNF formula 
  constructed from $\cL$, and $h$ be the pure Horn function it defines. Let $\Upsilon$ be a pure Horn 3-CNF 
  representation of $h$ obtained by some exact clause minimization algorithm when $\Phi$ is given as input.

  For convenience, let $\rho=\rho(\sigma)$ and $\zeta=\zeta(\sigma)=r/s=2^{O(\log^\beta\sigma)}/\sigma^{o(1)}$
  and recall Notation~\ref{nta:LC-sizes}.
  Substituting the quantities established in Lemma~\ref{lem:3cnf-reduction-size} into the bounds given by
  Corollary~\ref{cor:3cnf-upsilon-ub} and using and the new parametrization above (cf. 
  Remark~\ref{rem:new-LC-param}), we obtain that
  \begin{align*}
    |\Upsilon|_c & \geq t(\kappa(f)r + s) + \frac{d(\pi+2m\lambda'+2m-1)-1}{\lambda+\lambda'} \\
                 & \geq st(\kappa(f)r/s+1) + \Omega\bigl(\sigma^{1+o(1)} 2^{\log^\tau\sigma} \log^{2c\vartheta}\sigma\bigr)\\
                 & \geq st(\kappa(f)\zeta + 1) + \omega(s),
  \end{align*}
  and
  \begin{align*}
    |\Upsilon|_c & \leq t(\kappa(f)r + s) + d(\pi+m^2\lambda'+4m) \\
                 & \leq st(\kappa(f)r/s + 1) + O\bigl(\sigma^{4} 2^{\log^\tau\sigma} \log^3\sigma\bigr) \\ 
                 & \leq st(\kappa(f)\zeta + 1) + o(s^{5}),
  \end{align*}
  where $\tau>0$ is some constant. Similarly, for the number of variables we have that
  \begin{align*}
    t \leq |\Upsilon|_v & \leq t + dm(\lambda'+2)+m^2\lambda' \\
                        & \leq t + O\bigl(\sigma^{2} 2^{\log^{2\tau}\sigma} \log^2\sigma\bigr) \\
                        & \leq t + o(s^{3}).
  \end{align*}

  Now, choosing $\varepsilon'>0$ such that $t=s^{1/\varepsilon'}=\Omega(s^4)$ and supposing that $s\lra\infty$,
  we obtain the asymptotic expressions
  \[
    |\Upsilon|_c = s^{(1+1/\varepsilon')}\zeta(\sigma)\kappa(f)(1+o(1))
    \qquad\text{and}\qquad
    |\Upsilon|_v = s^{1/\varepsilon'}(1+o(1)).
  \]

  Bringing the existing gaps of the \LC promise instances into play, we then obtain the following dichotomy
  \begin{align*}
    \kappa(\cL)=1               & \Longrightarrow |\Upsilon|_c\leq s^{(1+1/\varepsilon')}\zeta(\sigma)(1+o(1)), \\
    \kappa(\cL)\geq\rho(\sigma) & \Longrightarrow |\Upsilon|_c\geq s^{(1+1/\varepsilon')}\zeta(\sigma)\rho(\sigma)(1+o(1)),
  \end{align*}
  with $\rho(\sigma)=\log^{\vartheta}\sigma$ in this case.
  Let $n=|\Upsilon|_v$ and let $\varepsilon=\varepsilon'/(1+o(1))$. Relating the number of clauses of $\Upsilon$ to the 
  number of variables of $h$, the above dichotomy reads as
  \begin{align*}
    \kappa(\cL)=1               & \Longrightarrow |\Upsilon|_c\leq n^{(1+\varepsilon')}\zeta(n^{\varepsilon})(1+o(1)), \\
    \kappa(\cL)\geq\rho(\sigma) & \Longrightarrow |\Upsilon|_c\geq n^{(1+\varepsilon')}\zeta(n^{\varepsilon})\rho(n^{\varepsilon})(1+o(1)),
  \end{align*}
  giving a hardness of approximation factor of $\rho(n^{\varepsilon})$ for the pure Horn 3-CNF clause minimization problem
  (cf. Theorem~\ref{thm:LC-Dinur-Harsha}). As $\varepsilon^\vartheta$ is a constant, it follows that the gap
  \[
  \rho(n^{\varepsilon})=\log^{\vartheta}n^{\epsilon}=O(\log^{\vartheta}n).
  \]
  Also, the number of clauses $n^{(1+\varepsilon')}\zeta(n^{\varepsilon})\leq n^{1+2\varepsilon'}$.
  
  To conclude the clause minimization part, observe that sub-exponential time in $\sigma$, namely, $2^{o(\sigma)}$ is equivalent 
  to $2^{o\bigl(n^{1/4-o(1)}\bigr)}$ time in $n$ and since $|\Upsilon|_c=O(n^{1+2\varepsilon'})$, the
  time bound follows for $\delta<(1-2\varepsilon')/4$.

  Regarding literal minimization, the structure of $\Phi$ implies its numbers of clauses and literals differ by only a constant
  (cf. Section~\ref{sec:3-cnf}) and therefore, similar results hold in this case as well.
\end{proof}

A natural next step consists in trying to push the hardness of approximation factor further by allowing 
super-polynomially sized constructions, that is, canonical formulae $\Phi$ whose number of clauses is
super-polynomial in $\sigma$. 

So let $b>1$ be such that $\alpha=bc>1$ and take $\rho(\sigma)=\log^\alpha\sigma$. We then have that 
\[ 
  \lambda =O\bigl(2^{\log^\alpha\sigma}\bigr)=O\bigl(\sigma^{\log^{\alpha-1}\sigma}\bigr)
  \qquad\text{and}\qquad
  \lambda'=O(\log^\alpha\sigma).
\]

Also, we have that 
\[
  \pi = O\Bigl(2^{\log^\alpha\sigma+O(\log^\beta\sigma)}\sigma\log^\alpha\sigma\Bigr),
\]
and that we should choose the value of the parameter $d$ such that
\[
  d = \Theta\left(\sigma 2^{\log^\alpha\sigma+O(\log^\beta\sigma)}\right).
\] 

Notice the above values imply that $\Phi$ is now super-polynomially sized in $\sigma$.

Following the steps of the proof of Theorem~\ref{thm:cnf-hardness-subexp}, we obtain that
\[
  \omega(s^2) \leq |\Upsilon'|_c - t(\kappa(f)r + s) \leq o\left(\sigma^{4+\log^{\alpha-1}\sigma}\right)
\]
and that
\[
  t \leq |\Upsilon'|_v \leq t + o\left(\sigma^{3+\log^{\alpha-1}\sigma}\right).
\]

Now, choosing $t=(8\sigma)^{\log^{\alpha-1}8\sigma}$ gives a hardness of approximation factor for pure Horn clause 
minimization equal to $\rho(\sigma)$. Writing $\sigma$ as a function of $n$, we obtain that $\sigma=(2^{\log^{1/\alpha}n})/8$ 
and hence, a hardness of approximation factor
\[
  \rho(n)=\log^\alpha\left(\frac{\displaystyle 2^{\sqrt[\alpha]{\log n}}}{8}\right)=\left(\sqrt[\alpha]{\log n}-3\right)^\alpha
         =O(\log n).
\]

As $|\Upsilon'|_c=O(n^2)$, sub-exponential time $2^{o(\sigma)}$ means
\[
  \mu(n):=2^{o\left(\left(2^{\sqrt[\alpha]{\log n}}/8\right)^{1/2}\right)} = o\left(2^{n^{1/(\log\log n)^C}}\right)
\]
time, for any (possibly large) constant $C\geq 1$. We then just proved the following theorem.

\begin{theorem}\label{thm:cnf-hardness-hyperpoly}
  Assuming the ETH, the minimum number of clauses and literals of pure Horn functions in $n$ variables cannot be 
  approximated in $\mu(n)$ time to within factors of $O(\log n)$, even when the input is restricted to 3-CNFs with 
  $O(n^2)$ clauses.
\end{theorem}

Pushing the approach further, if we take $\rho(\sigma)=2^{O(\log^\beta\sigma)}$ for some $0<\beta<1$, we have that 
\[
  \lambda= 2^{2^{O(\log^\beta\sigma)}}
  \qquad\text{and}\qquad
  \lambda'=2^{O(\log^\beta\sigma)}, 
\]
that $m=\sigma 2^{O(\log^\beta\sigma)}$, and that 
\[
  \pi=d=\sigma 2^{O(\log^\beta\sigma)} 2^{2^{O(\log^\beta\sigma)}},
\]
implying that $\Phi$ is now sub-exponentially sized in $\sigma$. We then have that 
\[
  \Omega\Big(2^{2^{O(\log^\beta\sigma)}}\Big)\leq \frac{|\Upsilon'|_c - t(\kappa(f)r + s)}{\sigma^2}\leq 
  O\Big(\sigma 2^{2^{O(\log^\beta\sigma)}}\Big)
\]
and that
\[
  t \leq |\Upsilon'|_v \leq t + O\Big(\sigma^2 2^{2^{O(\log^\beta\sigma)}}\Big).
\]

Now, letting $n:=|\Upsilon'|_v=t=2^{2^{1+\gamma\log^\beta\sigma}}$, for some constant $\gamma$, it gives
\[
  \sigma=2^{\left(\frac{1}{\gamma}\log\frac{\log n}{2}\right)^{1/\beta}}
\]
and a hardness of approximation factor $\rho(n)=(\log n)/2$, i.e., a similar $O(\log n)$ hardness factor 
under far more stringent time constraints. This last result is then rendered obsolete by 
Theorems~\ref{thm:3cnf-hardness} and \ref{thm:cnf-hardness-hyperpoly}.

\section{Conclusion}\label{sec:consider}
In the last three sections, we showed improved hardness of approximation results for the problems of
determining the minimum number of clauses and the minimum number of literals in prime pure Horn CNF and
3-CNF representations of pure Horn functions. In the polynomial time setting, we obtained a hardness of
approximation factor of $2^{\log^{1-o(1)} n}$ and when sub-exponential computational time is available,
we showed a factor of $O(\log^{\beta} n)$, where $\beta>0$ is some small constant and $n$ is the number of
variables of the pure Horn function. All these results hold even when the input CNF or 3-CNF formula is 
nearly linear, namely, when its size is $O(n^{1+\varepsilon})$ for some small constant $\varepsilon>0$. We also 
managed to obtain a factor of $O(\log n)$ under more stringent, albeit sub-exponential, time constraints
even when the input formula has size $O(n^2)$. In the polynomial time setting, our results are conditional
on the $\PT\neq\NP$ hypothesis, and are conditional on the Exponential Time Hypothesis (ETH) in the 
sub-exponential time scenario.

A natural question at this point concerns the tightness of our results. 
As mentioned in the introduction, Hammer and Kogan~\cite{HK93} showed that for a pure Horn function in 
$n$ variables, it is possible to approximate the minimum number of clauses and the minimum number of 
literals of a prime pure Horn CNF formula representing it to within factors of $n-1$ and $\binom{n}{2}$, 
respectively. 

Our results then leave a sub-exponential gap in the polynomial time setting, and we are not aware of the
existence of any sub-exponential time approximation algorithm for those problems. Naturally, narrowing 
or closing these gaps is highly desirable. One direction consists in designing new approximation algorithms 
with improved approximation guarantees. This does not seem to be an easy task to accomplish, nevertheless.

Another direction consists in further strengthening our hardness results. Replacing our constructions by
smaller (e.g. quasi-linear) ones will allow us to extend our sub-exponential hardness result to the scenario
where $\exp(o(n))$ computational time is available --- recall that our result is valid for $\exp(n^\delta)$
time, for some constant $0<\delta<1$, and that we conjecture that no constant approximation factor is possible
in the $\exp(o(n))$ time scenario. Besides barely improving the constants in our polynomial time proofs,
we believe this option has little, if anything else, to offer.

Improvements on two-query projection test, sub-constant error Probabilistically Checkable Proof (PCP) systems 
(cf. Arora and Safra~\cite{AS98}, Dinur and Safra~\cite{DS04}, Moshkovitz and Raz~\cite{MR-10-JACM}, and 
Dinur and Harsha~\cite{DH13}) might lead to larger gaps for the \LC (as a promise) problem, and as long as
the \LC instances are polynomially sized, our constructions would immediately imply larger hardness of 
approximation results for those problem. Moreover, improvements on \LC might also allow one to obtain 
hardness results when super-polynomial time is allowed, a case we left untreated.

As a third possibility, notice that it might be possible to start from a different (promise) problem and 
provide different constructions and different proofs. At the moment, it is not clear how to pursue this
venue. However, we conjecture that it is possible to improve the hardness of approximation factor for clause and 
literal minimization of a Horn function in $n$ variables to at least $O(n^\varepsilon)$, for some small 
$\varepsilon>0$. This conjecture concerns the polynomial time setting. In favor of it, we point out the 
fact that Umans~\cite{U99} showed that for general Boolean functions, the decision versions of clause and 
literal minimization problems are $\Sigma_2^p$-complete (namely, $\NP^\NP$-complete) and that it is $\Sigma_2^p$-hard 
to approximate such quantities to within factors of $N^\varepsilon$, where $N$ is the size of the input 
(a low-degree polynomial in the number of variables of the function in question) and $\varepsilon>0$ is
some small constant.

\textbf{Acknowledgements.}
  We would like to thank O. {\v{C}}epek and P. Ku{\v{c}}era for some discussions held during an early stage of 
  development of these results, and three anonymous referees for suggestions that improved the overall readability
  of this paper.


\end{document}